\documentclass[11pt]{article}
\usepackage{amsmath, amssymb, amsthm}
\usepackage[square,numbers]{natbib}
\usepackage{graphicx, verbatim}
\usepackage[margin=1in]{geometry}
\usepackage{url}
\usepackage{multirow}
\usepackage{rotating}
\usepackage{threeparttable}
\usepackage{mcode}
\usepackage{listings}
\usepackage{epsfig}
\usepackage{subfig}
\usepackage{bbm}
\usepackage{bm}
\usepackage[mathscr]{euscript}
\usepackage{float}
\usepackage{caption}
\usepackage{mathtools}

\numberwithin{equation}{section}
\newtheorem{thm}{Theorem}  
\newtheorem{prop}[thm]{Proposition}  

\newtheorem{defn}[thm]{Definition}

\title{General Price Bounds for Guaranteed Annuity Options}

\author{Raj Kumari Bahl and Sotirios Sabanis}

\date{\today}

\begin{document}

\maketitle

\begin{abstract}
\pagenumbering{roman}
     In this paper, we are concerned with the valuation of Guaranteed Annuity Options (GAOs) under the most generalised modelling framework where both interest and mortality rates are stochastic and correlated. Pricing these type of options in the correlated environment is a challenging task and no closed form solution exists in the literature. We employ the use of doubly stochastic stopping times to incorporate the randomness about the time of death and employ a suitable change of measure to facilitate the valuation of survival benefit, there by adapting the payoff of the GAO in terms of the payoff of a basket call option. We derive general price bounds for GAOs by utilizing a conditioning approach for the lower bound and arithmetic-geometric mean inequality for the upper bound. The theory is then applied to affine models to present some very interesting formulae for the bounds under the affine set up. Numerical examples are furnished and benchmarked against Monte Carlo simulations to estimate the price of a GAO for a variety of affine processes governing the evolution of mortality and the interest rate.

\bigskip


\bigskip


\end{abstract}

\section{Introduction}
\pagenumbering{arabic}
    In the present era when financial institutions are facing serious challenges in the advent of improving life expectancy, pricing of key products such as `Guaranteed Annuity Options' which involve survival benefit has gained a lot of momentum. It is the need of the hour to equip the longevity product designers with an insight to efficient pricing of these instruments. This involves designing an apparatus that provides state of art solutions to measure the random impulse of mortality, which indeed calls for looking at mortality in a stochastic sense. Till, very lately the conventional approach of actuaries consisted in treating mortality in a deterministic way in contrast to interest rates which were assumed to possess a stochastic nature. Post this came the era of the assumption that mortality evolves in a stochastic manner but is independent of interest rates (see for example \cite{Biffis1}). However, the latter assumption is also far from being realistic. This is because  both extreme mortality events such as catastrophes and pandemics as well as improving life expectancy go a long way in influencing the value of interest rate. While the former shows a stronger effect in a short term, the latter affects the financial market in a gradual manner. Interested readers can refer to \cite{Deelestra}, \cite{Liu1}, \cite{Liu2} and \cite{Jalen} and the references therein. To the best of our knowledge, \cite{Miltersen} were the first ones to introduce dependence between mortality and interest rates in the actuarial world. In the context of the real world, a study by \cite{Nicolini} to understand the relation between these two underlying risks demonstrates that the decline of interest rate in pre-industrial England was perhaps triggered by the decline of adult mortality at the end of the 17th century. More recently \cite{Dacorogna} examine correlation between mortality and market risks in periods of extremes such as a severe pandemic outbreak while \cite{Dacorogna2} explore existence of this dependence within the Feller process framework.

As remarked in the beginning of this section `The Life Expectancy Revolution' has pressurised social security programmes of various nations thereby triggering fiscal crisis for governments who find it hard to fulfill the needs of an ever growing aging population. The price for this imbalance affects the financial markets adversely leading to a downtrend in returns on investments. To take care of these issues, EU's Solvency II Directive has laid out new insurance risk management practices for capital adequacy requirements based on the assumption of dependence between financial markets and life/health insurance markets including the correlation between the two underpinning risks viz. interest rate and mortality (c.f. Quantitative Impact Study 5:Technical Specifications \cite{QIS5}).

In this paper, we consider the most generalised modelling framework where both interest and mortality risks are stochastic and correlated. In a set up similar to \cite{Biffis1}, we advocate the use of doubly stochastic stopping times to incorporate the randomness about the time of death.

We then utilize this set up and the theory of comonotonicity to devise model-independent price bounds for Guaranteed Annuity Options (GAOs). These are options embedded in certain pension policies that provide the policyholders the right to choose between a lump sum at time of retirement/maturity or to convert the proceeds into an annuity at a guaranteed rate. The reports of the Insurance Institute of London (1972) (c.f. \cite{IIL}) show that the origin of GAOs dates back to 1839. However these instruments came into the limelight in UK in the era of 1970-1990. In the advent of increased life expectancy, the research on pricing of GAOs has gained a lot of momentum as the underpricing of such guarantees has already caused serious solvency problems to insurers, for example in the UK, as an after effect of encashment of too many GAOs, the world's oldest life insurer - Equitable Life had to close to new business in 2000.

The existing literature in the direction of pricing of GAO's under the correlation assumption is very thin and only Monte Carlo estimation of the GAO price is available for sophisticated models (c.f. \cite{Deelestra}). But Monte Carlo method is generally extremely time consuming for complex models (c.f. \cite{Runhuan}). This article is a concrete step in the direction of pricing of GAOs under the correlation direction. It investigates the designing of price bounds for GAO's under the assumption of dependence between mortality and interest rate risks and provides a much needed confidence interval for the pricing of these options. Moreover the proposed bounds are model-free or general in the sense they are applicable for all kinds of models and in particular suitable for the affine set up. Keeping pace with the relevant literature (c.f. \cite{Liu2}, \cite{Deelestra}), we applied a change of probability measure with the `Survival Zero Coupon Bond' as num\'{e}raire for the valuation of the GAO. This change of measure facilitates computation and enhances efficiency (c.f. \cite{Liu1}).  The organization of the paper follows. In section 2 we introduce the market framework with the necessary notations. In section 3 we define GAOs and show that their payoff is similar to that of a basket option. This is followed by Section 4 which highlights the technicalities of affine processes. Sections 5 and 6 are the core sections which present details on finding lower and upper bounds for GAOs. In section 7 we present examples while numerical investigations in support of the developed theory appear in Section 8. Section 9 then concludes the paper.

\section{The Market Framework}
In this section, we introduce the necessary set up required to construct the mathematical interplay between financial market and the mortality model. We denote by $\mathbb{P}$, the physical world measure and we utilize the fact that in the absence of arbitrage, at least one equivalent martingale measure (EMM) $\mathbb{Q}$ exists. We consider a filtered probability space $\left(\Omega, \mathscr{F}, \mathbb{F}, \mathbb{P}\right)$ where $\mathbb{F}=\{\mathscr{F}_{t}\}_{t\geq 0}$ such that the filtration is large enough to support a  process $X$ in $\mathbb{R}^{k}$, representing the evolution of financial variables and a process $Y$ in $\mathbb{R}^{d}$, representing the evolution of mortality. We take as given an adapted short rate process $r=\{r_{t}\}_{t\geq 0}$ such that it satisfies the technical condition $\int_{0}^{t} r_{s}ds<\infty$ a.s. for all $t\geq 0$. The short rate process $r$ represents the continuously compounded rate of interest of a risk-less security. Moreover, we concentrate on an insured life aged $x$ at time 0, with random residual lifetime denoted by $\tau_{x}$ which is an $\mathscr{F}_{t}$-stopping time.

The filtration $\mathbb{F}$ includes knowledge of the evolution of all state variables up to each time t and of whether the policyholder has died by that time. More explicitly, we have:
\[
\mathscr{F}_{t}=\mathscr{G}_{t} \vee \mathscr{H}_{t}
\]
where
\[
\mathscr{G}_{t} \vee \mathscr{H}_{t}=\sigma\left(\mathscr{G}_{t} \cup \mathscr{H}_{t}\right)
\]
with
\[\mathscr{G}_{t}=\sigma\left(Z_{s}:\;0\leq s\leq t\right),\;\;\;\mathscr{H}_{t}=\sigma\left(\mathbbm{1}_{\{\tau\leq s\}}:\;0\leq s\leq t\right)
\]
and where $Z=\left(X,Y\right)$ is the joint state variables process in $\mathbb{R}^{k+d}$. Thus we have
\[
\mathscr{G}_{t}=\mathscr{G}_{t}^{X} \vee \mathscr{G}_{t}^{Y}.
\]
In fact $\mathbb{H}=\{\mathscr{H}_{t}\}_{t\geq 0}$ is the smallest filtration with respect to which $\tau$ is a stopping time. In other words $\mathbb{H}$ makes $\mathbb{F}$ the smallest enlargement of $\mathbb{G}=\{\mathscr{G}_{t}\}_{t\geq 0}$ with respect to which $\tau$ is a stopping time, i.e.,
\[
\mathscr{F}_{t}=\cap_{s>t}\mathscr{G}_{s} \vee \sigma\left(\tau\wedge s\right), \;\forall t.
\]
We may think of $\mathscr{G}_{t}$ as carrying information captured from medical/demographical data collected at population/ industry level and of $\mathscr{H}_{t}$ as recording the actual occurence of death in an insurance portfolio.

To make the set up more robust, we assume that $\tau_{x}$ is the first jump-time of a nonexplosive $\mathscr{F}_{t}$-counting process $N$ recording at each time $t \geq 0$ whether the individual has died $\left(N_{t} \neq 0\right)$ or not $(N_{t} = 0)$. The stopping time $\tau_{x}$ is said to admit an intensity $\mu_{x}$ if $N$ does, i.e. if $\mu_{x}$ is a non-negative predictable process such that $\int_{0}^{t}\mu_{x}\left(s\right)ds<\infty$ a.s. for all $t\geq 0$ and such that the compensated process $M = \{N_{t}-\int_{0}^{t}\mu_{x}\left(s\right)ds:t\geq 0\}$ is a local $\mathscr{F}_{t}$-martingale. Our next assumption is that $N$ is a doubly stochastic process or Cox Process driven by a subfiltration
$\mathscr{G}_{t}$ of $\mathscr{F}_{t}$, with $\mathscr{G}_{t}$-predictable intensity $\mu$. This implies that on any particular trajectory $t \mapsto \mu_{t}\left(\omega\right)$ of $\mu$, the counting process $N$ is a Poisson-inhomogeneous process with parameter $\int_{0}^{.}\mu_{s}\left(\omega\right)ds$, i.e., we have that for all $t \in \left[0,T\right]$ and non-negative integer $k$,
\begin{equation}\label{2.1abcde}
\mathbb{P}\left(N_{T}-N_{t} = k|\mathscr{F}_{t} \vee \mathscr{G}_{T}\right)=\frac{\left(\int_{t}^{T}\mu_{s}ds\right)^{k}}{k!}e^{-\int_{t}^{T}\mu_{s}ds}.
\end{equation}
The main reason for the consideration of a strict subfiltration $\mathscr{G}_{T}$ of $\mathscr{F}_{t}$ is that it provides enough information about the evolution of the intensity of mortality, i.e., about the likelihood of death happening, but not enough information about the actual occurrence of death. Such information is carried by the larger filtration $\mathscr{F}_{t}$, with respect to which $\tau$ is a stopping time. From \eqref{2.1abcde} by putting $k = 0$, we now proceed to compute the `probability of survival' up to time $T \geq t$, on the set $\{\tau > t\}$. Let $A$ be the event of no death in the interval $t \in \left[0,T\right]$, i.e., $A\equiv\{N_{T}-N_{t} = 0\}$, then the tower property of conditional expectation tells us that
\begin{eqnarray}\label{2.1abcdef}
\mathbb{P}\left(\tau > T|\mathscr{F}_{t}\right) & = & E[\mathbbm{1}_{A}|\mathscr{F}_{t}] \nonumber\\
& = & E\left[E\left(\mathbbm{1}_{A}|\mathscr{F}_{t} \vee \mathscr{G}_{T}\right)|\mathscr{F}_{t}\right] \nonumber\\
& = & E\left[\mathbb{P}\left(N_{T}-N_{t} = 0|\mathscr{F}_{t} \vee \mathscr{G}_{T}\right)|\mathscr{F}_{t}\right] \nonumber\\
& = & E\left[e^{-\int_{t}^{T}\mu_{s}ds}|\mathscr{F}_{t}\right].
\end{eqnarray}

In fact, we characterize the conditional law of $\tau$ in several steps. Given the non-negative $\mathscr{G}_{t}$-predictable process $\mu$ is satisfying $\int_{0}^{t}\mu_{x}\left(s\right)ds<\infty$ a.s. for all $t>0$, we consider an exponential random variable $\Phi$ with parameter 1, independent of $\mathscr{G}_{\infty}$ and define the random time of death $\tau$ as the first time when the process $\int_{0}^{t}\mu_{s}ds$ is above the random threshold $\Phi$, i.e.,
\begin{equation}\label{2.1abcdefg}
\tau \doteq \{t\in \mathbb{R}^{+}: \int_{0}^{t}\mu_{s}\left(s\right)ds \geq \Phi\}.
\end{equation}
It is evident from \eqref{2.1abcdefg} that $\{\tau > T\}=\{\int_{0}^{t}\mu_{s}ds < \Phi\}$, for $T\geq 0$. Next, we work out $\mathbb{P}\left(\tau > T|\mathscr{G}_{t}\right)$ for $T\geq t\geq 0$ by using tower property of conditional expectation, independence of $\Phi$ and $\mathscr{G}_{\infty}$ and facts that $\mu$ is a $\mathscr{G}_{t}$-predictable process and $\Phi \sim Exponential\left(1\right)$, i.e.,
\begin{equation}\label{2.2ab}
\mathbb{P}\left(\tau > T|\mathscr{G}_{t}\right) =  E\left[e^{-\int_{0}^{T}\mu_{s}ds}|\mathscr{G}_{t}\right].
\end{equation}
In fact, the same result holds for $0 \leq T < t$. Further, we observe that $\{\tau > t\}$ is an atom of $\mathscr{H}_{t}$. As a result, in a manner similar to \cite{Biffis1}, we have constructed a doubly stochastic $\mathscr{F}_{t}$-stopping time driven by $\mathscr{G}_{t} \subset \mathscr{F}_{t}$ in the following way (c.f. \cite{billie}, ex 34.4, p.455):
\begin{eqnarray}\label{2.2abc}
\mathbb{P}\left(\tau > T|\mathscr{G}_{T} \vee \mathscr{F}_{t}\right) & = & \mathbbm{1}_{\{\tau > t\}}E\left[\mathbbm{1}_{\{\tau > T\}}|\mathscr{G}_{T} \vee \mathscr{H}_{t}\right] \nonumber\\
& = & \mathbbm{1}_{\{\tau > t\}}e^{-\int_{t}^{T}\mu_{s}ds}.
\end{eqnarray}
Next, the conditioning on $\mathscr{F}_{t}$ can be replaced by conditioning on $\mathscr{G}_{t}$ as shown in the Appendix C of \cite{Biffis1}.

We remark that, we do not take $\mathscr{G}_{t} \vee \sigma(\Phi)$ as our filtration $\mathscr{G}_{t}$ because, in that case, the stopping time $\tau$ would be predictable and would not admit an intensity. The construction potrayed here guarantees that $\tau$ is a totally inaccessible stopping time, a concept intuitively meaning that the insured’s death arrives as a total surprise to the insurer (see \cite{Protter}, Chapter III.2, for details). With this, we move to the focal point of this paper viz. GAOs.

\section{Guaranteed Annuity Options}
\subsection{Introduction}
A Guaranteed Annuity Option(GAO) is a contract that gives the policyholder the flexibility to convert his/her survival benefit into an annuity at a pre-specified conversion rate. The guaranteed conversion rate denoted by $g$, can be quoted as an annuity/cash value ratio. According to \cite{Bolton}, the most popular choice for for the guaranteed conversion rate $g$ for males aged 65 in the UK in the 1980s was $g=\frac{1}{9}$, which means that per \pounds 1000 cash value can be converted into an annuity of \pounds 111 per annum. The GAO would have a positive value if the guaranteed conversion rate is higher than the available conversion rate; otherwise the GAO is worthless since the policyholder could use the cash to obtain higher value of annuity from the primary market. As a result, the moneyness of the GAO at maturity depends on the price of annuity available in the market at that time and this in turn is calculated using the prevailing interest and mortality rates.

\subsection{Mathematical Formulation}
Consider an $x$ year old policyholder at time 0 who has an access to a unity amount at his retirement age $R_{x}$. Then, a GAO gives the policyholder a choice to choose at time $T=R_{x}-x$ between an annual payment of $g$ or a cash payment of 1. Let $\ddot{a}_{x}\left(T\right)$ denote a whole life annuity due for a person aged $x$ at time 0, which gives an annual payment of one unit amount at the start of each year, this payment beginning from time $T$ and conditional on survival. If $w$ is the largest possible survival age then we have

\begin{eqnarray}\label{5.1}
\ddot{a}_{x}\left(T\right) & = & \sum_{j=0}^{w-\left(T+x\right)-1}\mathbb{E}\left[e^{-\int_{T}^{T+j}\left(r_{s}+\mu_{s}\right)ds}|\mathscr{G}_{T}\right] \nonumber\\
& = & \sum_{j=0}^{w-\left(T+x\right)-1} \tilde{P}\left(T,T+j\right),
\end{eqnarray}
where
\begin{equation}\label{4.3}
\tilde{P}\left(t,T\right)=\mathbb{E}\left[e^{-\int_{t}^{T}\left(r_{s}+\mu_{s}\right)ds}|\mathscr{G}_{t}\right]
\end{equation}
denotes the price at time t of a pure endowment insurance with maturity $T$ for an insured of age $x$ at time 0 who is still alive at time $t$. This insurance instrument is nomenclated as a \emph{survival zero-coupon bond} abbreviated as SZCB by \cite{Deelestra} and the authors remark that it can be used as a num\`eraire because it can be replicated by a strategy that involves longevity bonds (c.f. \cite{Lin2}) in analogy with the usual bootstrapping methodology used to find the zero rate curve starting by coupon bonds. This insurance instrument pays one unit of money at time $T$ upon the survival of the insured at that time. In fact $r+\mu$ can be viewed as a fictitious short rate or yield to compare these instruments with their financial counterparts.

At time $T$, the value of the contract having the above embedded GAO can be described by the following decomposition
\begin{eqnarray}\label{5.2}
V\left(T\right)& = & \max(g \ddot{a}_{x}\left(T\right), 1) \nonumber\\
& = & 1 + g \max\left(\ddot{a}_{x}\left(T\right)-\frac{1}{g}\right).
\end{eqnarray}

In order to apply risk neutral evaluation, we state a result from \cite{Biffis1} to compute the fair values of a basic payoff involved by standard insurance contracts. These are benefits, of amount possibly linked to other security prices, contingent on survival over a given time period. We require the short rate process $r$ and the intensity of mortality $\mu$ to satisfy the technical conditions stated in Section 2.

\begin{prop}(Survival benefit). Let $C$ be a bounded $\mathscr{G}_{t}$-adapted process. Then, the time-$t$ fair value $SB_{t}\left(C_{T} ; T\right)$ of the time-$T$ survival benefit of amount $C_{T}$, with $0 \leq t \leq T$ , is given by:
\begin{equation}\label{3.1ab}
SB_{t}\left(C_{T} ; T\right) = \mathbb{E}\left[e^{-\int_{t}^{T}r_{s}ds}\mathbbm{1}_{\{\tau > T\}}C_{T}|\mathscr{F}_{t}\right]=\mathbbm{1}_{\{\tau > t\}}\mathbb{E}\left[e^{-\int_{t}^{T}\left(r_{s}+\mu_{s}\right)ds}C_{T}|\mathscr{G}_{t}\right]
\end{equation}
In particular, if $C$ is $\mathscr{G}_{t}^{X}$-adapted and $X$ and $Y$ are independent, then, the following holds
\begin{equation}\label{3.1abc}
SB_{t}\left(C_{T} ; T\right) =\mathbbm{1}_{\{\tau > t\}}\mathbb{E}\left[e^{-\int_{t}^{T}r_{s}ds}C_{T}|\mathscr{G}_{t}^{X}]\mathbb{E}[e^{-\int_{t}^{T}\mu_{s}ds}|\mathscr{G}_{t}^{Y}\right]
\end{equation}
\end{prop}

\begin{proof}
A comprehensive proof can be found in \cite{Biffis1}.
\end{proof}

Thus, we have the value at time $t = 0$ of the second term in \eqref{5.3}, which is called the GAO option price entered by an $x$-year policyholder at time $t = 0$ as
\begin{equation}\label{5.3}
C(0, x, T )=\mathbb{E}\left[e^{-\int_{0}^{T}\left(r_{s}+\mu_{s}\right)ds}g\left(\ddot{a}_{x}\left(T\right)-\frac{1}{g}\right)^{+}\right].
\end{equation}
In order to facilitate calculation, we adopt the following change of measure.

\subsection{Change of Measure}
We advocate a change of measure similar to the one adopted in \cite{Deelestra}. We define a new probability measure $\tilde{Q}$ with the Radon-Nikodym derivative of $\tilde{Q}$ w.r.t $\mathbb{Q}$ as:
\begin{equation}\label{4.1}
\frac{d\tilde{Q}}{d\mathbb{Q}} := \eta_{T}=\frac{e^{-\int_{0}^{T}\left(r_{s}+\mu_{s}\right)ds}}{\mathbb{E}\left[e^{-\int_{0}^{T}\left(r_{s}+\mu_{s}\right)ds}\right]}
\end{equation}
where $\mathbb{E}$ denotes the usual expectation w.r.t the EMM $\mathbb{Q}$ and we will use $\tilde{E}$ to denote the expectation w.r.t the new probability measure $\tilde{Q}$. Further on using Bayes' Rule for conditional expectation, the survival benefit in \eqref{3.1ab} can be rewritten as
\begin{equation}\label{4.2}
SB_{t}\left(C_{T} ; T\right) = \mathbbm{1}_{\{\tau > t\}}\tilde{P}\left(t,T\right)\tilde{E}\left[C_{T}|\mathscr{G}_{t}\right]
\end{equation}

The advantage of the change of measure approach is that the complex expectation appearing in the survival benefit given in \eqref{3.1ab} has been decomposed into two simpler expectations: the first one corresponds to the price of the SZCB given in \eqref{4.3} and the second one is connected to the expected value of the survival benefit $C_{T}$ under the new probability measure $\tilde{Q}$ which needs to be determined. In the passing, one notes that in \eqref{4.2} if $C_{T}=1$, we get a very interesting relationship
\begin{equation}\label{4.4}
SB_{t}\left(1 ; T\right) = \mathbbm{1}_{\{\tau > t\}}\tilde{P}\left(t,T\right).
\end{equation}
In particular
\begin{equation}\label{4.5}
SB_{0}\left(1 ; T\right) = \mathbbm{1}_{\{\tau > t\}}\tilde{P}\left(0,T\right).
\end{equation}

A similar change of measure has been employed by \cite{Liu1} and \cite{Liu2} with the only difference that they use the unitary survival benefit given in \eqref{4.4} as the num\`eraire. On the contrary, \cite{Jalen} have used a twin change of measure to compute value of a GAO.

\subsection{Payoff}

Under the new probability measure $\tilde{Q}$ defined in \eqref{4.1}, the GAO option price decomposes into the following product
\begin{equation}\label{5.4}
C(0, x, T ) = g \tilde{P}\left(0,T\right)\tilde{E}\left[\left(\ddot{a}_{x}\left(T\right)-\frac{1}{g}\right)^{+}\right]
\end{equation}
where $\tilde{P}\left(0,T\right)$ is defined in \eqref{4.3}. To develop ideas further, we express the payoff in a more appealing form as follows:
\begin{equation}\label{5.5}
C(0, x, T ) = g \tilde{P}\left(0,T\right){\displaystyle\tilde{E}}\left[\left(\sum_{i=1}^{n-1}S_{T}^{\left(i\right)}-\left(K-1\right)\right)^{+}\right]
\end{equation}
where we utilize the fact that $\tilde{P}\left(T,T\right)=1$ and define $n=w-\left(T+x\right)$ and
\begin{equation}\label{5.6}
S_{T}^{\left(i\right)}=\tilde{P}\left(T,T+i\right);\;i=1,2,...,n-1.
\end{equation}
The last term on the R.H.S in the payoff of the GAO resembles the payoff of a basket option having unit weights and the SZCBs, maturing at times $T+1, T+2,...,w-x-1$ acting as the underlying assets.
We seek to evaluate tight model-independent bounds for the GAOs in the ensuing sections. To the best of our knowledge, the equations \eqref{5.3} and \eqref{5.4} have only been valued by Monte Carlo simulations for specific choice of models. In \cite{Liu1}, numerical experiments in the Gaussian setting have shown that \eqref{5.4} is a little bit more precise and in particular it is less time consuming than the implementation of \eqref{5.3}. \cite{Deelestra} have investigated these calculations for different affine models such as the multi-CIR and the Wishart cases. \cite{Liu2} have computed very specific comonotonic bounds for GAOs in the Gaussian framework.

\section{Affine Processes}
Affine processes are essentially Markov processes with conditional characteristic function of the affine form. A thorough discussion of these processes on canonical state space appears in \cite{Duffie2} and \cite{Filipovic}. More recently the development of multivariate stochastic volatility models has lead to the evolution of applications of affine processes on non-canonical state spaces, in particular on the cone of positive semi-definite matrices. A plethora of research papers are available to explore and interested readers can refer to \cite{Cuchiero1} for details. A unified approach on affine processes is presented in \cite{Keller2} and following this approach we recall the details of the affine processes in the Appendix A. In regards to the evolution of interest rates and the force of mortality we consider a set up similar to \cite{Deelestra}.

Suppose we have a time-homogeneous affine Markov process $X$ taking values in a non-empty convex subset $E$ of $\mathbb{R}^{d}$,  $\left(d\geq 1\right)$ equipped with the inner product $\langle\cdot,\cdot\rangle$. We then assume that the dynamics of the interest rate and force of mortality are given respectively as follows.
\begin{equation}\label{51.11}
r_{t}=\bar{r}+\langle R,X_{t}\rangle
\end{equation}
and
\begin{equation}\label{51.12}
\mu_{t}=\bar{\mu}+\langle M,X_{t}\rangle
\end{equation}
where $\bar{r}, \bar{\mu} \in \mathbb{R}$, $M, R \in \mathbb{R}_{d} \mbox{ or } M_{d}$ where $M_{d}$ is the set of real square matrices of order $d$.

This means that the interest rate and mortality are linear projections of the common stochastic factor $X$ along constant directions given by the parameter $R$ and $M$ respectively. We will be interested in the cases where the $X$ is a classical affine process on the state space $\mathbb{R}^{m}_{+}\times \mathbb{R}^{n}$ or an affine Wishart process on the state space $S_{d}^{+}$, which is the set of $d\times d$ symmetric positive definite matrices. The inner product possesses the flexibility to condense into scalar product or trace depending on the nature of $R$ and $M$ being respectively vectors or matrices. In the former set up we consider multi-dimensional CIR case (c.f. \cite{CIR}). In the case of Vasicek model (c.f. \cite{Vasicek}), the affine set up is uni-dimensional. A very good reference to show that the stochastic processes underlying the Vasicek and CIR models fall under the affine set up is \cite{martink7}.

In the passing it is important to note that the affiness of the underlying model is preserved as we move from the physical world to the the risk neutral environment, although new affine dynamics emerge (c.f. \cite{Biffis2} and \cite{Duffie1}). In fact, more recently \cite{Dhaene} examine the conditions under which it is possible or not to translate the independence assumption from the physical world to the pricing world.

We now state without proof the following proposition which presents the methodology to value SZCBs and in turn GAOs. A detailed proof appears in \cite{Gnoatto1} and the necessary notations are defined in the Appendix A.
\begin{prop}\label{411.7}
Let $X$ be a conservative affine process on $S_{d}^{+}$ under the risk neutral measure $\mathbb{Q}$. Let the short rate be given in accordance with \eqref{51.11}. Let $\tau=T-t$, then the price of a zero-coupon bond is given by
{
\allowdisplaybreaks
\begin{eqnarray}\label{51.122}
\tilde{P}\left(t, T\right) & = & \mathbb{E}\left[e^{-\int^{T}_{t}\left(\bar{r}+\bar{\mu}+\langle R+M,X_{u}\rangle\right) du}|\mathscr{F}_{t}\right]\nonumber\\
& = & e^{-\left(\bar{r}+\bar{\mu}\right)\tau}e^{-\tilde{\phi}\left(\tau,R+M\right)-\langle \tilde{\psi}\left(\tau,R+M\right),X_{t}\rangle},
\end{eqnarray}
}
where $\tilde{\phi}$ and $\tilde{\psi}$ satisfy the following Ordinary Differential Equations (ODEs) which are known also as Riccati ODE's.
\begin{equation}\label{51.13}
\frac{\partial\tilde{\phi}}{\partial \tau}=\tilde{\Im}\left(\tilde{\psi}\left(\tau,R+M\right)\right),\;\;\tilde{\phi}\left(0,R+M\right)=0,
\end{equation}
\begin{equation}\label{51.14}
\frac{\partial\tilde{\psi}}{\partial\tau}=\tilde{\Re}\left(\tilde{\psi}\left(\tau,R+M\right)\right),\;\;\tilde{\psi}\left(0,R+M\right)=0,
\end{equation}
with
\begin{equation}\label{51.15}
\tilde{\Im}\left(\tilde{\psi}\left(\tau,R+M\right)\right)=\langle b,\tilde{\psi}\left(\tau,R+M\right)\rangle-\int_{S_{d}^{+}\setminus \{0\}}\left(e^{-\langle \tilde{\psi}\left(\tau,R+M\right),\xi\rangle} -1\right) m\left(d\xi\right),
\end{equation}
and
{
\allowdisplaybreaks
\begin{eqnarray}\label{51.16}
\tilde{\Re}\left(\tilde{\psi}\left(\tau,R+M\right)\right) & = & -2\tilde{\psi}\left(\tau,R+M\right)\alpha \tilde{\psi}\left(\tau,R+M\right)+B^{T}\left(\tilde{\psi}\left(\tau,R+M\right)\right)\nonumber\\
& {} & {} -\int_{S_{d}^{+}\setminus \{0\}}\left(\frac{e^{-\langle \tilde{\psi}\left(\tau,R+M\right),\xi\rangle} -1+\langle \chi\left(\xi\right),\tilde{\psi}\left(\tau,R+M\right)\rangle}{\parallel \xi \parallel^{2} \wedge 1}\right)\mu\left(d\xi\right)+R+M.\nonumber\\
\end{eqnarray}
}
\end{prop}

In fact it is interesting to note that assuming this kind of affine structure means that our fictitious yield model is ``affine"  in the sense that there is, for each maturity $T$, an affine mapping $Z_{T}:\mathbb{R}^{n}\rightarrow \mathbb{R}$ such that, at any time $t$, the yield of any SZCB of maturity $T$ is $Z_{T}\left(X_{t}\right)$ echoing the results obtained in the seminal paper of \cite{Duffie0}.

As a result we have for $i=1,2,...,n-1$,
\begin{equation}\label{51.123}
S_{T}^{\left(i\right)} =  e^{-\left(\bar{r}+\bar{\mu}\right)i}e^{-\tilde{\phi}\left(i,R+M\right)-\langle \tilde{\psi}\left(i,R+M\right),X_{T}\rangle},
\end{equation}
where $\tilde{\phi}\left(i,R+M\right)$ and $\tilde{\psi}\left(i,R+M\right)$ satisfy the equations \eqref{51.13} and \eqref{51.14} with $\tau=i$. Alternatively, one may write
\begin{equation}\label{51.123a}
S_{T}^{\left(i\right)} =  S_{0}^{\left(i\right)}e^{X_{T}^{\left(i\right)}};\;i=1,2,...,n-1,
\end{equation}
with
\begin{equation}\label{51.123b}
S_{0}^{\left(i\right)} =  e^{-\left(\left(\bar{r}+\bar{\mu}\right)i+\tilde{\phi}\left(i,R+M\right)\right)}
\end{equation}
and
\begin{equation}\label{51.123c}
X_{T}^{\left(i\right)} =  -\langle \tilde{\psi}\left(i,R+M\right),X_{T}\rangle.
\end{equation}
As a result in the affine case, by using equation \eqref{51.123} in \eqref{5.5} the formula for GAO payoff can be written in a very compact form as shown below.
\begin{equation}\label{51.124}
C(0, x, T ) = g \tilde{P}\left(0,T\right){\displaystyle\tilde{E}}\left[\left(\sum_{i=1}^{n-1}e^{-\left(\bar{r}+\bar{\mu}\right)i}e^{-\tilde{\phi}\left(i,R+M\right)-\langle \tilde{\psi}\left(i,R+M\right),X_{T}\rangle}-\left(K-1\right)\right)^{+}\right],
\end{equation}
where $\tilde{P}\left(0,T\right)$ given by equation \eqref{51.122} with $\tau=T$. As a result in the affine case, our quest of bounds for the GAO becomes simplified as we are dealing only with $X_{T}$.

The analytical tractability of affine processes is essentially linked to generalized Riccati equations as given above which can be in general solved by numerical methods although explicit solutions are available in the Vasicek (c.f. \cite{Vasicek}) and CIR (c.f. \cite{CIR}) models without jumps.

\section{Lower Bound for Guaranteed Annuity Options}
We now proceed to work out appropriate lower bounds for the payoff of the GAO as given in \eqref{5.5}. Invoking Jensen's inequality , we have
\begin{eqnarray}\label{6.1}
\tilde{E}\ensuremath{\left[\left({\displaystyle \sum_{i=1}^{n-1}}S_{T}^{\left(i\right)}-\left(K-1\right)\right)^{+}\right]}  & \geq & \tilde{E}\ensuremath{\left[\left({\displaystyle \sum_{i=1}^{n-1}}\tilde{E}\left(S_{T}^{\left(i\right)}|\Lambda\right)-\left(K-1\right)\right)^{+}\right]}.
\end{eqnarray}
The general derivation concerning lower bounds for stop loss premium of a sum of random variables based on Jensen's inequality can be found in \cite{Simon} and for its application to Asian basket options, one can refer to \cite{Deelstra2}. Define
\begin{equation}\label{6.1a}
S=\displaystyle \sum_{i=1}^{n-1}S_{T}^{\left(i\right)}
\end{equation}
and
\begin{equation}\label{6.1b}
S^{l}=\displaystyle \sum_{i=1}^{n-1}\tilde{E}\left(S_{T}^{\left(i\right)}|\Lambda\right)
\end{equation}
Thus, we have obtained
\begin{equation}\label{6.1c}
S \geq_{cx} S^{l}.
\end{equation}
Now, suitably tailoring the inequality \eqref{6.1}, we obtain
\begin{equation}\label{6.2}
C(0, x, T ) \geq g \tilde{P}\left(0,T\right)\tilde{E}\ensuremath{\left[\left({\displaystyle \sum_{i=1}^{n-1}}\tilde{E}\left(S_{T}^{\left(i\right)}|\Lambda\right)-\left(K-1\right)\right)^{+}\right]}.
\end{equation}

\subsection{A Lower Bound}
In case, if the random variable $\Lambda$ is independent of the prices of pure endowments having term periods $1,2,...,n-1$ at the time $T$, i.e., of $S_{T}^{\left(i\right)};\;i=1,2,...,n-1$, respectively, the bound in \eqref{6.2} simply reduces to:
\begin{equation}\label{6.9}
C(0, x, T )\geq g \tilde{P}\left(0,T\right)\tilde{E}\ensuremath{\left[\left({\displaystyle \sum_{i=1}^{n-1}}\tilde{E}\left(S_{T}^{\left(i\right)}\right)-\left(K-1\right)\right)^{+}\right]}.
\end{equation}
or even more precisely as the outer expectation is redundant, we obtain a very trivial bound for GAO expressed in terms of expectation of $S_{T}^{i}$, i.e.,
\begin{equation}\label{6.10}
C(0, x, T )\geq g \tilde{P}\left(0,T\right)\ensuremath{\left({\displaystyle \sum_{i=1}^{n-1}}\tilde{E}\left(S_{T}^{\left(i\right)}\right)-\left(K-1\right)\right)^{+}}=:\mbox{ GAOLB}.
\end{equation}

\subsubsection{The Lower Bound under the Affine Set Up}
Under the affine set up of section 4 (c.f. equation \eqref{51.123}), the lower bound given in equation \eqref{6.10} reduces to
\begin{equation}\label{6.10a}
\mbox{ GAOLB}^{\emph{aff}}=g \tilde{P}\left(0,T\right)\ensuremath{\left({\displaystyle \sum_{i=1}^{n-1}}\left(e^{-\left(\left(\bar{r}+\bar{\mu}\right)i+\tilde{\phi}\left(i,R+M\right)\right)} \mathscr{L}\left(\tilde{\psi}\left(i,R+M\right)\right)\right)-\left(K-1\right)\right)^{+}}
\end{equation}
where $\mathscr{L}$ denotes the Laplace transform of $X_{T}$ with parameter $\tilde{\psi}\left(i,R+M\right)$ under the transformed measure $\tilde{Q}$. This means that if one can lay hands on the distribution of $X_{T}$, this bound has a very compact form.

\section{Upper Bound for Guaranteed Annuity Options}
In order to obtain an upper bound for GAOs which is directly applicable to the affine set up, we make use of arithmetic-geometric mean inequality in a manner similar to \cite{Grasselli5} who used this methodology to arrive at an upper bound for basket options.

Let us first define the arithmetic and geometric mean of the $\left(n-1\right)$ pure endowments appearing in the payoff of GAO (c.f. \eqref{5.5}) respectively as
\begin{equation}\label{20.18}
A_{T}^{\left(n-1\right)}=\frac{1}{n-1} {\displaystyle \sum_{i=1}^{n-1}} S_{T}^{\left(i\right)}
\end{equation}
and
\begin{equation}\label{20.19}
G_{T}^{\left(n-1\right)}= \left({\displaystyle \prod_{i=1}^{n-1}} S_{T}^{\left(i\right)}\right)^{\frac{1}{n-1}},
\end{equation}
where $S_{T}^{\left(i\right)};\;i=1,2,...,n-1$ are defined in equation \eqref{5.6}. It is well known that
\begin{equation}\label{20.19a}
A_{T}^{\left(n-1\right)} \geq G_{T}^{\left(n-1\right)}\;\;a.s.
\end{equation}
Further, let us define the log-geometric average as
\begin{equation}\label{20.20}
Y_{T}^{\left(n-1\right)}=\frac{1}{n-1} {\displaystyle \sum_{i=1}^{n-1}} \ln S_{T}^{\left(i\right)}.
\end{equation}
Next we define as in equation \eqref{51.123a},
\begin{equation}\label{20.20a}
X_{T}^{\left(i\right)}=\ln \left(\frac{S_{T}^{\left(i\right)}}{S_{0}^{\left(i\right)}}\right);\;i=1,2,...,n-1.
\end{equation}
Further, we assume that the joint characteristic function of $\left(X_{T}^{\left(1\right)},...,X_{T}^{\left(n-1\right)}\right)$ can be obtained under the transformed measure $\tilde{Q}$, where we define
\begin{equation}\label{20.20b}
\phi_{T}\left(\boldsymbol{\gamma}\right)=\tilde{E}\left[e^{i\sum_{k=1}^{n-1}\gamma_{k}X_{T}^{\left(k\right)}}\right]
\end{equation}
with $\boldsymbol{\gamma}=\left[\gamma_{1}, \gamma_{2},...,\gamma_{n-1}\right]$. As the next step, we obtain the relationship between log-geometric average and $X_{T}^{\left(i\right)}$'s as follows
\begin{eqnarray}\label{20.20c}
Y_{T}^{\left(n-1\right)} & = & \frac{1}{n-1} {\displaystyle \sum_{i=1}^{n-1}} \ln \left(\frac{S_{T}^{\left(i\right)}}{S_{0}^{\left(i\right)}}S_{0}^{\left(i\right)}\right)\nonumber\\
& = & \frac{1}{n-1} {\displaystyle \sum_{i=1}^{n-1}} X_{T}^{\left(i\right)}+ Y_{0}^{\left(n-1\right)}.
\end{eqnarray}
Next, we try to express the characteristic function of log-geometric average under the transformed measure $\tilde{Q}$ in terms of the joint characteristic function of $X_{T}^{\left(i\right)}$'s viz. $\phi_{T}\left(\boldsymbol{\gamma}\right)$ defined in equation \eqref{20.20b}. Let $\phi_{Y_{T}}\left(\gamma_{0}\right)$ denote the characteristic function of log-geometric average $Y_{T}^{\left(n-1\right)}$ with parameter $\gamma_{0}$. Then we have
\begin{eqnarray}\label{20.20d}
\phi_{Y_{T}}\left(\gamma_{0}\right) & = & \tilde{E}\left[e^{i\gamma_{0}Y_{T}^{\left(n-1\right)}}\right]\nonumber\\
& = & \tilde{E}\left[e^{i\gamma_{0}Y_{0}^{\left(n-1\right)}+i\sum_{k=1}^{n-1}\left(\frac{\gamma_{0}}{n-1}\right)X_{T}^{\left(k\right)}}\right]\nonumber\\
& = & e^{i\gamma_{0}Y_{0}^{\left(n-1\right)}}\phi_{T}\left(\frac{\gamma_{0}}{n-1}\mathbf{1}\right)
\end{eqnarray}
where $\mathbf{1}=\left(1,1,...,1\right)$ is a $1\times \left(n-1\right)$ vector of 1's, so that  $\frac{\gamma_{0}}{n-1}\mathbf{1}$ is $1\times \left(n-1\right)$ vector with components $\frac{\gamma_{0}}{n-1}$ and $\phi_{T}\left(\boldsymbol{\gamma}\right)$ is defined in \eqref{20.20b}.
In light of equation \eqref{20.18}, we can express the GAO payoff formula given in equation \eqref{5.5} as
\begin{equation}\label{20.21}
C(0, x, T ) = g \left(n-1\right) \tilde{P}\left(0,T\right)\tilde{E}\left[\left(A_{T}^{\left(n-1\right)}-K'\right)^{+}\right],
\end{equation}
where
\begin{equation}\label{20.22}
K'=\frac{K-1}{n-1}.
\end{equation}
Adding and subtracting $G_{T}^{\left(n-1\right)}$ within the $max$ function on R.H.S. of equation \eqref{20.21}, and exploiting equation \eqref{20.19a}, we obtain an upper bound of GAO as
\begin{eqnarray}\label{20.27}
C(0, x, T ) & \leq & g \left(n-1\right) \tilde{P}\left(0,T\right)\left(\tilde{E}\left[\left(G_{T}^{\left(n-1\right)}-K^{'}\right)^{+}\right]+\tilde{E}\left[A_{T}^{\left(n-1\right)}\right]-\tilde{E}\left[G_{T}^{\left(n-1\right)}\right]\right) \nonumber\\
& = & :\mbox{ GAOUB}
\end{eqnarray}
We make use of Fourier inversion to compute the call type expectation involved in the upper bound
and we state the result in the following proposition.
\begin{prop}\label{203}
Given the geometric mean of $n-1$ pure endowments defined in equation \eqref{20.19} and $K^{'}>0$,
\begin{equation}\label{20.28}
\tilde{E}\left[\left(G_{T}^{\left(n-1\right)}-K^{'}\right)^{+}\right]=\frac{e^{-\delta \ln K'}}{\pi}\int_{0}^{\infty}e^{-i\eta \ln K^{'}}\Psi_{T}^{G}\left(\eta;\delta\right)d\eta
\end{equation}
where $\Psi_{T}^{G}\left(\eta;\delta\right)$ denotes the Fourier transform of $\tilde{E}\left[\left(G_{T}^{\left(n-1\right)}-K^{'}\right)^{+}\right]$ with respect to $\ln K^{'}$ along with the damping factor $e^{\delta \ln K^{'}}$ such that
\begin{equation}\label{20.28a}
\Psi_{T}^{G}\left(\eta;\delta\right)=e^{i\left(\eta-i\left(\delta+1\right)\right)Y_{0}^{\left(n-1\right)}}\frac{\phi_{T}\left(\frac{\eta-i\left(\delta+1\right)}{n-1}\mathbf{1}\right)}{\delta^2+\delta-\eta^{2}+i\eta\left(2\delta+1\right)},
\end{equation}
where the parameter $\delta$ tunes the damping factor (c.f. \cite{Carr} and \cite{Grasselli5}) and $\phi_{T}\left(.\right)$ is defined in equation \eqref{20.20b}.
\end{prop}

\begin{proof}
Let $f_{Y_{T}}\left(y\right)$ denote the probability density function (p.d.f.) of the log-geometric average $Y_{T}^{\left(n-1\right)}$. We introduce the damping factor in accordance with \cite{Carr}. Then, by definition, the Fourier transform of $\tilde{E}\left[\left(G_{T}^{\left(n-1\right)}-K^{'}\right)^{+}\right]$ with respect to $\ln K^{'}$ along with the damping factor $e^{\delta \ln K^{'}}$ is given as
\begin{eqnarray}\label{20.29}
\Psi_{T}^{G}\left(\eta;\delta\right) & = & \int_{\mathbb{R}}e^{i\eta\ln K^{'}+\delta \ln K^{'}}\tilde{E}\left[\left(e^{Y_{T}^{\left(n-1\right)}}-K^{'}\right)^{+}\right]d\ln K^{'} \nonumber\\
& = & \int_{\mathbb{R}}e^{i\eta\ln K^{'}+\delta \ln K^{'}}\int_{\ln K^{'}}^{\infty}\left(e^{y}-K^{'}\right)f_{Y_{T}}\left(y\right)\;dy\;d\ln K^{'}\nonumber\\
& = & \int_{\mathbb{R}}e^{i\eta\ln K^{'}+\delta \ln K^{'}}\int_{\ln K^{'}}^{\infty}e^{y}f_{Y_{T}}\left(y\right)\;dy\;d\ln K^{'}\nonumber\\
& {} {}  & {} -\int_{\mathbb{R}}e^{i\eta\ln K^{'}+\delta \ln K^{'}}\int_{\ln K^{'}}^{\infty}K^{'}f_{Y_{T}}\left(y\right)\;dy\;d\ln K^{'}\nonumber\\
& = & \Psi_{T}^{G_{1}}\left(\eta;\delta\right)-\Psi_{T}^{G_{2}}\left(\eta;\delta\right).
\end{eqnarray}
We evaluate both integrals by adopting a change of order of integration, as detailed below
\begin{eqnarray}\label{20.30}
\Psi_{T}^{G_{1}}\left(\eta;\delta\right) & = & \int_{\mathbb{R}}e^{y}\left(\int_{-\infty}^{y}e^{i\eta\ln K^{'}+\delta \ln K^{'}}d\ln K^{'}\right)f_{Y_{T}}\left(y\right)dy \nonumber\\
& = & \frac{1}{i\eta+\delta}\int_{\mathbb{R}}e^{i\left(\eta-i\left(\delta+1\right)\right)y}f_{Y_{T}}\left(y\right)dy \nonumber\\
& = & \frac{\phi_{Y_{T}}\left(\eta-i\left(\delta+1\right)\right)}{i\eta+\delta}\nonumber\\
& = & e^{i\left(\eta-i\left(\delta+1\right)\right)Y_{0}^{\left(n-1\right)}}\frac{\phi_{T}\left(\frac{\eta-i\left(\delta+1\right)}{n-1}\mathbf{1}\right)}{i\eta+\delta}.
\end{eqnarray}
where the last couple of statements follow from the definition of the characteristic function of $Y_{0}^{\left(n-1\right)}$ given in \eqref{20.20d} and its link to the joint characteristic function of joint characteristic function of $\left(X_{T}^{\left(1\right)},...,X_{T}^{\left(n-1\right)}\right)$ defined in \eqref{20.20b}. On the same lines we have
\begin{equation}\label{20.31}
\Psi_{T}^{G_{2}}\left(\eta;\delta\right) = e^{i\left(\eta-i\left(\delta+1\right)\right)Y_{0}^{\left(n-1\right)}}\frac{\phi_{T}\left(\frac{\eta-i\left(\delta+1\right)}{n-1}\mathbf{1}\right)}{i\eta+\left(\delta+1\right)}.
\end{equation}
Substituting $\Psi_{T}^{G_{1}}\left(\eta;\delta\right)$ and $\Psi_{T}^{G_{2}}\left(\eta;\delta\right)$ in equation \eqref{20.29}, remembering the damping factor we get the requisite result given in equation \eqref{20.28}.
\end{proof}
In a similar manner we obtain
\begin{equation}\label{20.32a}
\tilde{E}\left[G_{T}^{\left(n-1\right)}\right]=e^{Y_{0}^{\left(n-1\right)}}\phi_{T}\left(\frac{-i}{n-1}\mathbf{1}\right).
\end{equation}
We then plug the formulae \eqref{20.28} and \eqref{20.32a} into equation \eqref{20.27} to obtain the upper bound $\mbox{ GAOUB}$.


\subsection{The Upper Bound under the Affine Set Up}
Consider the affine set up of section 4 (c.f. equations \eqref{51.123}-\eqref{51.123c}). Let $\phi_{X_{T}}$ denote the characteristic function of $X_{T}$ with parameter $\Lambda$ under the transformed measure $\tilde{Q}$ so that
\begin{equation}\label{20.33a}
\phi_{X_{T}}\left(\Lambda\right)=\tilde{E}\left[e^{i\langle \Lambda,\;X_{T} \rangle}\right].
\end{equation}
Now using equation \eqref{51.123c}, we see that the joint characteristic function of $\left(X_{T}^{\left(1\right)},...,X_{T}^{\left(n-1\right)}\right)$ under the transformed measure $\tilde{Q}$, given in equation \eqref{20.20b} becomes ,
\begin{equation}\label{20.33}
\phi_{T}^{\emph{aff}}\left(\boldsymbol{\gamma}\right)= \phi_{X_{T}}\left(-\sum_{k=1}^{n-1}\gamma_{k}\tilde{\psi}\left(k,R+M\right)\right),
\end{equation}
where  $\left(-\sum_{k=1}^{n-1}\gamma_{k}\tilde{\psi}\left(k,R+M\right)\right)$ is the parameter of the characteristic function, with $\tilde{\psi}\left(k,R+M\right)$ satisfying the equations \eqref{51.14} with $\tau=k$. As a result, $\Psi_{T}^{G}\left(\eta;\delta\right)$ given in equation \eqref{20.28a} can be written in a more compact way as
\begin{equation}\label{20.34}
\Psi_{T}^{G^{\emph{aff}}}\left(\eta;\delta\right)=e^{i\left(\eta-i\left(\delta+1\right)\right)Y_{0}^{\left(n-1\right)}}\frac{\phi_{X_{T}}\left(-\frac{\left(\eta-i\left(\delta+1\right)\right)}{n-1}\sum_{k=1}^{n-1}\tilde{\psi}\left(k,R+M\right)\right)}{\delta^2+\delta-\eta^{2}+i\eta\left(2\delta+1\right)}.
\end{equation}
Similarly, we have from equation \eqref{20.32a},
\begin{equation}\label{20.35}
\tilde{E}^{\emph{aff}}\left[G_{T}^{\left(n-1\right)}\right]=e^{Y_{0}^{\left(n-1\right)}}\phi_{X_{T}}\left(\frac{i}{n-1}\sum_{k=1}^{n-1}\tilde{\psi}\left(k,R+M\right)\right).
\end{equation}
Moreover, using the definition of arithmetic average given in equation \eqref{20.18} and utilizing \eqref{51.123}, we see that
\begin{equation}\label{20.36}
\tilde{E}^{\emph{aff}}\left[A_{T}^{\left(n-1\right)}\right]=\frac{1}{n-1}\sum_{k=1}^{n-1}\left(e^{-\left(\left(\bar{r}+\bar{\mu}\right)k+\tilde{\phi}\left(k,R+M\right)\right)} \mathscr{L}\left(\tilde{\psi}\left(k,R+M\right)\right)\right),
\end{equation}
where as defined in Section 5.1.1, $\mathscr{L}$ denotes the Laplace transform of $X_{T}$ with parameter $\tilde{\psi}\left(k,R+M\right)$ under the transformed measure $\tilde{Q}$. Finally we substitute equation \eqref{20.34} in the expression \eqref{20.28} and then the result and the equations \eqref{20.35}-\eqref{20.36} into \eqref{20.27} to obtain
\begin{eqnarray}\label{20.37}
\mbox{ GAOUB}^{\emph{aff}}& = & g \left(n-1\right) \tilde{P}\left(0,T\right)\Bigg(\frac{1}{n-1}\sum_{k=1}^{n-1}\left(e^{-\left(\left(\bar{r}+\bar{\mu}\right)k+\tilde{\phi}\left(k,R+M\right)\right)} \mathscr{L}\left(\tilde{\psi}\left(k,R+M\right)\right)\right)\nonumber\\
& {} &  -e^{Y_{0}^{\left(n-1\right)}}\phi_{X_{T}}\left(\frac{i}{n-1}\sum_{k=1}^{n-1}\tilde{\psi}\left(k,R+M\right)\right)\nonumber\\
& & +\frac{e^{-\delta \ln K'}}{\pi}\int_{0}^{\infty}\frac{e^{-i\left(\eta \ln K^{'}-\left(\eta-i\left(\delta+1\right)\right)Y_{0}^{\left(n-1\right)}\right)}}{\delta^2+\delta-\eta^{2}+i\eta\left(2\delta+1\right)}\phi_{X_{T}}\left(-\frac{\left(\eta-i\left(\delta+1\right)\right)}{n-1}\sum_{k=1}^{n-1}\tilde{\psi}\left(k,R+M\right)\right)d\eta\Bigg),\nonumber\\
\end{eqnarray}
where $\phi_{X_{T}}\left(.\right)$ is defined in equation \eqref{20.33a} and $\mathscr{L}$ denotes the Laplace transform of $X_{T}$ under the transformed measure $\tilde{Q}$.

\section{Examples}
We now derive lower and upper bounds by choosing specific models for the interest rate and force of mortality.
\subsection{The Multi-CIR Model}

We now consider a $p$-dimensional affine process $X:=\left(X_{t}\right)_{t\geq 0}$ having independent components $\left(X_{it}\right)_{t\geq 0}$ that function according to the following CIR risk-neutral dynamics:
\begin{equation}\label{50.1}
dX_{it}=k_{i}\left(\theta_{i}-X_{it}\right)dt+\sigma_{i}\sqrt{X_{it}}dW^{\mathbb{Q}}_{it}, \;i=1,...,p.
\end{equation}
One can refer to \cite{Deelestra} to show that this model fits into the general affine framework.

\subsubsection{Survival Zero Coupon Bond Pricing}
Adhering to the notations of the affine set-up defined in section 6, in context of mortality and interest rate, let $M, R \in \mathbb{R}_{n}$ with respective components $M_{i}, R_{i};\;i=1,2,...,p$. The price of a zero-coupon bond under the multi CIR model \eqref{51.18a} is given by
{
\allowdisplaybreaks
\begin{eqnarray}\label{50.2}
\tilde{P}\left(t, T\right) & = & \mathbb{E}\left[e^{-\int^{T}_{t}\left(\bar{r}+
\bar{\mu}\right)+\langle \left(R+M\right),X_{s}\rangle ds}|\mathscr{F}_{t}\right]\nonumber\\
& = & e^{-\left(\bar{r}+\bar{\mu}\right)\left(T-t\right)}{\displaystyle \prod_{i=1}^{p}} \mathbb{E}\left[e^{-\int^{T}_{t}\langle \left(R_{i}+M_{i}\right),X_{is}\rangle ds}|\mathscr{F}_{t}\right]\nonumber\\
& = & e^{-\left(\bar{r}+\bar{\mu}\right)\left(T-t\right)}{\displaystyle \prod_{i=1}^{p}} e^{-\tilde{\phi}_{i}\left(T-t,R_{i}+M_{i}\right)-\tilde{\psi}_{i}\left(T-t,R_{i}+M_{i}\right)X_{it}}
\end{eqnarray}
}
where $\tilde{\phi}_{i}$ and $\tilde{\psi}_{i}$ satisfy the following Riccatti equations for every $i=1,2,...,p$ (c.f. \cite{Duffie1}):
\begin{equation}\label{50.3}
\begin{cases}
\frac{\partial\tilde{\psi}\left(\tau,u_{i}\right)}{\partial \tau} =  1-k_{i}\tilde{\psi}_{i}\left(\tau,u_{i}\right)+\frac{u_{i}\sigma^{2}_{i}}{2}\tilde{\psi}_{i}\left(\tau,u_{i}\right)^{2},\\
\frac{\partial\tilde{\phi}\left(\tau,u_{i}\right)}{\partial \tau} = k_{i}\theta_{i}u_{i}\tilde{\psi}_{i}\left(\tau,u_{i}\right),
\end{cases}
\end{equation}

with $\tau=T-t$, $u_{i}=R_{i}+M_{i}$ and initial conditions $\tilde{\psi}_{i}\left(0,u_{i}\right)=0$ and $\tilde{\phi}_{i}\left(0,u_{i}\right)=0$.

The solutions of this system with $i=1,2,...,p$ are

\begin{eqnarray}\label{50.4}
\tilde{\psi}_{i}\left(\tau,u_{i}\right) & = & \frac{2u_{i}}{\eta\left(u_{i}\right)+ k_{i}}-\frac{4u_{i}+\eta\left(u_{i}\right)}{\eta\left(u_{i}\right)+ k_{i}}\nonumber\\
& {} & \times \frac{1}{\left(\eta\left(u_{i}\right)+ k_{i}\right)\exp\left[\eta\left(u_{i}\right)\tau\right]+\eta\left(u_{i}\right)-k_{i}}
\end{eqnarray}

\begin{eqnarray}\label{50.5}
\tilde{\phi}_{i}\left(\tau,u_{i}\right) & = & \frac{k_{i}\theta_{i}}{\sigma^{2}_{i}}\left[\eta\left(u_{i}\right)k_{i}\right]\tau\nonumber\\
& {} & +\frac{2k_{i}\theta_{i}}{\sigma^{2}_{i}}\log\left[\left(\eta\left(u_{i}\right)+ k_{i}\right)\exp\left[\eta\left(u_{i}\right)\tau\right]+\eta\left(u_{i}\right)-k_{i}\right]\nonumber\\
& {} & -\frac{2k_{i}\theta_{i}}{\sigma^{2}_{i}}\log\left(2\eta\left(u_{i}\right)\right)
\end{eqnarray}
where $\eta\left(u_{i}\right)=\sqrt{k^{2}_{i}+2u_{i}\sigma^{2}_{i}}$.

\subsubsection{Price of the GAO}
We use equations \eqref{51.124} and \eqref{50.2} to obtain the price of the GAO under the transformed measure $\tilde{Q}$ as
\begin{equation}\label{50.6}
C(0, x, T ) = g \tilde{P}\left(0,T\right){\displaystyle\tilde{E}}\left[\left(\sum_{i=1}^{n-1}e^{-\left(\bar{r}+\bar{\mu}\right)i}{\displaystyle \prod_{j=1}^{p}} e^{-\tilde{\phi}_{j}\left(i,R_{j}+M_{j}\right)-\tilde{\psi}_{j}\left(i,R_{j}+M_{j}\right)X_{jT}}-\left(K-1\right)\right)^{+}\right]
\end{equation}
where $\tilde{P}\left(0,T\right)$ given by equation \eqref{50.2} with $\tau=T$ while $\tilde{\psi}_{j}\left(i,R_{j}+M_{j}\right)$ and $\tilde{\phi}_{j}\left(i,R_{j}+M_{j}\right)$ are given by equations \eqref{50.4} and \eqref{50.5}.

\subsubsection{Distribution of $X_{T}$}
In order to obtain explicit bounds for the GAO in the multidimensional CIR case, we need to obtain the distribution of $X_{jT}$ under the transformed measure $\tilde{Q}$. We state this in the following proposition (c.f. \cite{CIR} and \cite{Deelestra} for details).

\begin{prop}\label{4111.2}
The dynamics of the CIR process $X_{jt}$ defined in equation \eqref{50.1} under the transformed measure $\tilde{Q}$ are given by
\begin{equation}\label{50.7}
dX_{jt}=k_{j}^{'}\left(\theta_{j}^{'}-X_{jt}\right)dt+\sigma_{j}\sqrt{X_{jt}}dW^{'}_{jt}, \;j=1,...,p.
\end{equation}
where
\begin{equation}\label{50.8}
k_{j}^{'}=k_{j}+\sigma_{j}^{2}\tilde{\psi}_{j}\left(0,R_{j}+M_{j}\right),
\end{equation}
\begin{equation}\label{50.9}
\theta_{j}^{'}=\frac{k_{j}\theta_{j}}{k_{j}+\sigma_{j}^{2}\tilde{\psi}_{j}\left(0,R_{j}+M_{j}\right)}
\end{equation}
and $X_{j0};\;j=1,2,...,p$ is the initial value of the process. Then the density function of $X_{jT}$ is given by
\begin{equation}\label{50.10}
f_{X_{jT}}\left(x\right)=f_{\frac{\chi^{2}\left(\nu_{jT},\lambda_{jT}\right)}{c_{jT}}}\left(x\right)=c_{jT}f_{\chi^{2}\left(\nu_{jT},\;\lambda_{jT}\right)}\left(c_{jT}x\right)
\end{equation}
where $f_{\chi^{2}\left(\nu_{jT},\;\lambda_{jT}\right)};\;j=1,2,...,p$ is the p.d.f. of non-central $\chi^{2}$ with degrees of freedom $\nu_{j}$ and non-centrality parameter $\lambda_{jT}$ such that
\begin{equation}\label{50.11}
c_{jT}=\frac{4k_{j}^{'}}{\sigma_{j}^{2}\left(1-e^{-k_{j}^{'}T}\right)},
\end{equation}
\begin{equation}\label{50.12}
\nu_{jT}=\frac{4k_{j}^{'}\theta_{j}^{'}}{\sigma_{j}^{2}}
\end{equation}
and
\begin{equation}\label{50.13}
\lambda_{jT}=c_{jT}X_{j0}e^{-k_{j}^{'}T}.
\end{equation}
\end{prop}
The moment generating function (m.g.f.) of $X_{jT}$ has a very interesting exposition as detailed below (c.f. \cite{Daniel} for details).
\begin{equation}\label{50.13a}
\mathscr{M}_{X_{jT}}\left(s_{j}\right)=\left(\beta\left(s_{j}\right)\right)^{\bar{\nu_{j}}}e^{\lambda^{'}_{jT}\left(\beta\left(s_{j}\right)-1\right)}
\end{equation}
where
\begin{equation}\label{50.13b}
\beta\left(s_{j}\right)=\left(1-s_{j}\mu_{jT}\right)^{-1},
\end{equation}
with
\begin{equation}\label{50.13c}
\mu_{jT}=\frac{2}{c_{jT}},
\end{equation}
\begin{equation}\label{50.13d}
\bar{\nu_{j}}=\frac{\nu_{jT}}{2}
\end{equation}
and
\begin{equation}\label{50.13e}
\lambda^{'}_{jT}=2\lambda_{jT}.
\end{equation}

\subsubsection{The Lower Bound $\mbox{GAOLB}^{\left(MCIR\right)}$}
The lower bound $\mbox{GAOLB}$ obtained in equation \eqref{6.10} condenses into a very compact formula for the Multi-CIR case in a manner similar to the formula \eqref{6.10a} under the affine set up. Before unravelling the same, we see that in light of the notations defined in Section 4, one can write
\begin{equation}\label{50.5a}
S_{T}^{\left(i\right)}=S_{0}^{\left(i\right)}e^{X_{T}^{\left(i\right)}};\;i=1,2,...,n-1,
\end{equation}
where
\begin{equation}\label{50.5b}
S_{0}^{\left(i\right)}=e^{-\left(\left(\bar{r}+\bar{\mu}\right)i+\sum_{j=1}^{p}\tilde{\phi}_{j}\left(i,R_{j}+M_{j}\right)\right)}
\end{equation}
and
\begin{equation}\label{50.5c}
X_{T}^{\left(i\right)}=-{\displaystyle \sum_{j=1}^{p}}\tilde{\psi}_{j}\left(i,R_{j}+M_{j}\right)X_{jT},
\end{equation}
where $\tilde{\phi}_{j}\left(i,R_{j}+M_{j}\right)$ and $\tilde{\psi}_{j}\left(i,R_{j}+M_{j}\right)$ for $i=1,2,...,n-1$ and $j=1,2,...,p$ are given in equations \eqref{50.4}-\eqref{50.5} with $\tau$ replaced by $i$. Further, $X_{jT};\;j=1,2,...,p$ are independent random variables and their m.g.f. is given in equation \eqref{50.13a}. This leads us to the formulation of the lower bound for the Multi-CIR case presented in the form of the following proposition:
\begin{prop}
The lower bound under the multi-CIR case is
\begin{eqnarray}\label{50.14}
\mbox{ GAOLB}^{MCIR} & = & g \tilde{P}\left(0,T\right)\Bigg({\displaystyle \sum_{i=1}^{n-1}}\Bigg(e^{-\left(\left(\bar{r}+\bar{\mu}\right)i+\sum_{j=1}^{p}\tilde{\phi}_{j}\left(i,R_{j}+M_{j}\right)\right)+\sum_{j=1}^{p}\lambda^{'}_{jT}\left(\beta\left(-\tilde{\psi}_{j}\left(i,R_{j}+M_{j}\right)\right)-1\right)}\nonumber\\
& {} {} {} {} & {} {} \;\;\;\; \;\;\;\;\;\;\;\;\;\;\left(\prod_{j=1}^{p} \left(\beta\left(-\tilde{\psi}_{j}\left(i,R_{j}+M_{j}\right)\right)\right)^{\bar{\nu_{j}}}\right)\Bigg)-\left(K-1\right)\Bigg)^{+}
\end{eqnarray}
where $\beta\left(.\right)$ is defined in \eqref{50.13b} and $\bar{\nu_{j}}$ and $\lambda^{'}_{jT}$ are given in equations \eqref{50.13d}-\eqref{50.13e} and $\tilde{\psi}_{j}\left(i,R_{j}+M_{j}\right)$ for $i=1,2,...,n-1;j=1,2,...,p$ are given in \eqref{50.5}.
\end{prop}

\begin{proof}
Using the formula for lower bound given in equation \eqref{6.10}
\begin{equation}\label{50.15}
\mbox{ GAOLB} = g \tilde{P}\left(0,T\right)\ensuremath{\left({\displaystyle \sum_{i=1}^{n-1}}\tilde{E}\left(S_{T}^{\left(i\right)}\right)-\left(K-1\right)\right)^{+}}
\end{equation}
Using the formula of $S_{T}^{\left(i\right)}$ given in equations \eqref{50.5a}-\eqref{50.5c} we have
\begin{equation}\label{50.16}
\mbox{ GAOLB}^{MCIR}=g \tilde{P}\left(0,T\right)\ensuremath{\left({\displaystyle \sum_{i=1}^{n-1}}\left(e^{-\left(\left(\bar{r}+\bar{\mu}\right)i+\sum_{j=1}^{p}\tilde{\phi}_{j}\left(i,R_{j}+M_{j}\right)\right)}\prod_{j=1}^{p} \mathscr{M}_{X_{jT}}\left(-\tilde{\psi}_{j}\left(i,R_{j}+M_{j}\right)\right)\right)-\left(K-1\right)\right)^{+}}.
\end{equation}
Using the definition of m.g.f. of $X_{jT};\;j=1,2,...,p$ given in equation \eqref{50.13a} we obtain the requisite result.
\end{proof}
\subsubsection{The Upper Bound $\mbox{GAOUB}^{\left(MCIR\right)}$}
Under the formulation of the pure endowments constituting the GAO basket under the MCIR case (\eqref{50.5a}-\eqref{50.5c}), we write
\begin{equation}\label{50.17}
Y_{0}^{\left(n-1\right)}=-\frac{\left(\bar{r}+\bar{\mu}\right)n}{2}-\frac{1}{n-1} {\displaystyle \sum_{k=1}^{n-1}} \sum_{j=1}^{p}\tilde{\phi}_{j}\left(k,R_{j}+M_{j}\right)
\end{equation}
using the definition of log-geometric average $Y_{T}^{\left(n-1\right)}$ given in equation \eqref{20.20} in Section 6.3. We now exploit the set up of the upper bound under the affine case given in Section 6.3.1 and note that here instead of $\tilde{\phi}\left(k,R+M\right)$ and $\tilde{\psi}\left(k,R+M\right)$, we have respectively $\sum_{j=1}^{p}\tilde{\phi}_{j}\left(k,R_{j}+M_{j}\right)$ and $\sum_{j=1}^{p}\tilde{\psi}_{j}\left(k,R_{j}+M_{j}\right)$ since we are dealing with a $p$-dimensional CIR process.
Thus the joint characteristic function of $\left(X_{T}^{\left(1\right)},...,X_{T}^{\left(n-1\right)}\right)$ under the transformed measure $\tilde{Q}$, given in equation \eqref{20.20b} becomes ,
\begin{equation}\label{50.18}
\phi_{T}^{\emph{MCIR}}\left(\boldsymbol{\gamma}\right)= \prod_{j=1}^{p}\phi_{X_{jT}}\left(-\sum_{k=1}^{n-1}\gamma_{k}\tilde{\psi}_{j}\left(k,R_{j}+M_{j}\right)\right),
\end{equation}
where $\boldsymbol{\gamma}=\left[\gamma_{1}, \gamma_{2},...,\gamma_{n-1}\right]$, $\phi_{X_{jT}};\;j=1,2,...,p$ denotes the characteristic function of the $X_{jT}$ with parameter $\left(-\sum_{k=1}^{n-1}\gamma_{k}\tilde{\psi}_{j}\left(k,R_{j}+M_{j}\right)\right)$ for $j=1,2,...,p$, with $\tilde{\psi}_{j}\left(k,R_{j}+M_{j}\right)$ for $k=1,2,...,n-1$; $j=1,2,...,p$ are given in equation \eqref{50.5} with $\tau$ replaced by $k$. $\phi_{X_{jT}}\left(s\right)$ can be obtained from the formula of its m.g.f. given in equation \eqref{50.13a} by replacing $s$ by $is$. Further, we see that $\Psi_{T}^{G}\left(\eta;\delta\right)$ given in equation \eqref{20.28a} reduces to
\begin{equation}\label{50.19}
\Psi_{T}^{G^{\emph{MCIR}}}\left(\eta;\delta\right)=e^{i\left(\eta-i\left(\delta+1\right)\right)Y_{0}^{\left(n-1\right)}}\frac{\prod_{j=1}^{p}\phi_{X_{jT}}\left(-\frac{\left(\eta-i\left(\delta+1\right)\right)}{n-1}\sum_{k=1}^{n-1}\tilde{\psi}_{j}\left(k,R_{j}+M_{j}\right)\right)}{\delta^2+\delta-\eta^{2}+i\eta\left(2\delta+1\right)}.
\end{equation}
Next, we obtain $\tilde{E}^{\emph{MCIR}}\left[G_{T}^{\left(n-1\right)}\right]$ from equation \eqref{20.34} by utilizing \eqref{50.18}. Further, we compute
\begin{equation}\label{50.20}
\tilde{E}^{\emph{MCIR}}\left[A_{T}^{\left(n-1\right)}\right]=\frac{1}{n-1}\sum_{k=1}^{n-1}\left(e^{-\left(\left(\bar{r}+\bar{\mu}\right)k+\sum_{j=1}^{p}\tilde{\phi}_{j}\left(k,R_{j}+M_{j}\right)\right)} \prod_{j=1}^{p}\mathscr{M}_{X_{jT}}\left(-\tilde{\psi}_{j}\left(k,R_{j}+M_{j}\right)\right)\right).
\end{equation}
Finally we plug in the components one by one into equation \eqref{20.27} to obtain the upper bound $\mbox{GAOUB}^{\left(MCIR\right)}$.

\subsection{The Wishart Short Rate Model}
\subsubsection{The Set Up}
In this section, we assume that the affine process $X:=\left(X_{t}\right)_{t\geq 0}$ is a d-dimensional Wishart process. Given a $d \times d$ matrix Brownian motion $W$ (i.e a matrix whose entries are independent Brownian motions) the Wishart process $X$ (without jumps) is defined as the solution of the $d \times d$-dimensional stochastic differential equation
\begin{equation}\label{51.17}
dX_{t}=\left(\beta Q^{T}Q+HX_{t}+X_{t}H^{T}\right)dt+\sqrt{X_{t}}dW_{t}Q+Q^{T}dW^{T}_{t}\sqrt{X_{t}}, \;t\geq 0,
\end{equation}
where $X_{0}=x \in S_{d}^{+}$, $\beta \geq d-1$, $H \in M_{d}$, $Q \in GL_{d}$ and $Q^{T}$ denotes its transpose. $M_{d}$ has been defined in Section 4 while $GL_{d}$ denote the set of invertible real $d \times d$ matrices In short, we assume that the law of $X$ is $WIS_{d}\left(x_{0},\beta,H,Q\right)$.

\subsubsection{Existence and Uniqueness of Solution}
This process was pioneered by \cite{Bru2} and she showed the existence and uniqueness of a weak solution for Eq. \eqref{51.17}. She also established the existence of a unique strong solution taking values in $S_{d}^{++}$, i.e. the interior of the cone of positive semi-definite symmetric $d\times d$ matrices that we have denoted by $S_{d}^{+}$.

\subsubsection{Generator}
\cite{Bru2} has calculated the infinitesimal generator of the Wishart process as:
\begin{equation}\label{51.18}
\mathscr{A}=Tr\left(\left(\beta Q^{T}Q+Hx+xH^{T}\right)D^{S}+2xD^{S}Q^{T}QD^{S}\right),
\end{equation}
where $Tr$ stands for trace and $D^{S}=\left(\partial/\partial x_{ij}\right)_{1\leq i, j\leq d}$. A good reference for understanding the detailed derivation of this generator is \cite{Alfonsi} and following this reference we have defined generator in Appendix A.

\subsubsection{Survival Zero Coupon Bond Pricing}
We now give an explicit formula for calculating the the survival zero coupon bond price under the Wishart short rate model.

\begin{thm}\label{411.8}
Let the dynamics for short rate and mortality rate be given in accordance with equation \eqref{51.11} respecively as
\begin{equation}\label{51.18a}
r_{t}=\bar{r}+Tr\left[RX_{t}\right]
\end{equation}
and
\begin{equation}\label{51.18b}
\mu_{t}=\bar{\mu}+Tr\left[MX_{t}\right]
\end{equation}
for a process $X$ with law $WIS_{d}\left(x_{0},\beta,H,Q\right)$. Let $R, M \in S_{d}^{++}$ and $\tau=T-t$, then the price of a zero-coupon bond under the Wishart short rate model \eqref{51.18a} is given by
{
\allowdisplaybreaks
\begin{eqnarray}\label{51.19}
\tilde{P}\left(t, T\right) & = & \mathbb{E}\left[e^{-\int^{T}_{t}\left(\bar{r}+\bar{\mu}+Tr\left[\left(R+M\right)X_{u}\right]\right) du}|\mathscr{F}_{t}\right]\nonumber\\
& = & e^{-\tilde{\phi}\left(\tau,R+M\right)-Tr\left[\tilde{\psi}\left(\tau,R+M\right)X_{t}\right]},
\end{eqnarray}
}
where $\tilde{\phi}$ and $\tilde{\psi}$ satisfy the following system of ODEs:
\begin{equation}\label{51.20}
\begin{cases}
\frac{\partial\tilde{\phi}}{\partial \tau} =  Tr\left[\beta Q^{T}Q\tilde{\psi}\left(\tau,R+M\right)\right]+\bar{r}+\bar{\mu},\\
\tilde{\phi}\left(0,R+M\right) = 0,\\
\frac{\partial\tilde{\psi}}{\partial \tau} = \tilde{\psi}\left(\tau,R+M\right)H+H^{T}\tilde{\psi}\left(\tau,R+M\right),\\
\;\;\;\;\;\;-2\tilde{\psi}\left(\tau,R+M\right)Q^{T}Q\tilde{\psi}\left(\tau,R+M\right)+R+M,\\
\tilde{\psi}\left(0,R+M\right) = 0.\end{cases}
\end{equation}
\end{thm}

\begin{proof}
Consider the expectation in equation \eqref{51.19}. As remarked in section 2, the conditioning on $\mathscr{F}_{t}$ can be reduced to that on $\mathscr{G}_{t}$ and so we define $t\leq T$, define
\begin{equation}\label{51.20a}
F\left(t, X_{t}\right) = f\left(\tau, X_{t}\right) = \mathbb{E}\left[e^{-\int^{T}_{t}\left(\bar{r}+\bar{\mu}+Tr\left[\left(R+M\right)X_{u}\right]\right) du}|X_{t}\right].
\end{equation}
This conditional expectation is the Feynman-Kac representation which satisfies the following Partial Differential Equation (PDE):
\begin{equation}\label{51.20b}
\begin{cases}
\frac{\partial f\left(\tau, x\right)}{\partial\tau}  =  \mathscr{A}f\left(\tau, x\right)-\left(\bar{r}+\bar{\mu}+Tr\left[\left(R+M\right)x\right]\right)f\left(\tau, x\right),\\
f\left(0, x\right)=1,\end{cases}
\end{equation}
for all $x \in S_{d}^{+}$, where $\mathscr{A}$ is the infinitesimal generator of the Wishart process given in equation \eqref{51.18}. We introduce a candidate solution given below
\begin{equation}\label{51.23}
f\left(\tau, x\right)=e^{-\tilde{\phi}\left(\tau,R+M\right)-Tr\left[\tilde{\psi}\left(\tau,R+M\right)x\right]}
\end{equation}
so that
\begin{equation}\label{51.24}
\frac{\partial f\left(\tau, x\right)}{\partial\tau}=\left(-\frac{\partial\tilde{\phi}}{\partial \tau}-Tr\left[\frac{\partial\tilde{\psi}}{\partial \tau}x\right]\right)f\left(\tau, x\right)
\end{equation}
Also it is clear that
\begin{equation}\label{51.25}
\mathscr{A}e^{-\tilde{\phi}\left(\tau,R+M\right)-Tr\left[\tilde{\psi}\left(\tau,R+M\right)x\right]}=e^{-\tilde{\phi}\left(\tau,R+M\right)}\mathscr{A}e^{-Tr\left[\tilde{\psi}\left(\tau,R+M\right)x\right]},
\end{equation}
where on using the generator of the Wishart process given in equation \eqref{51.18}, we have
\begin{eqnarray}\label{51.26}
\mathscr{A}e^{-Tr\left(\tilde{\psi}\left(\tau,R+M\right)x\right)} & = & \Bigg(-Tr\left[\beta Q^{T}Q\tilde{\psi}\left(\tau,R+M\right)\right]+Tr\Bigg[\Bigg(2\tilde{\psi}\left(\tau,R+M\right)Q^{T}Q\tilde{\psi}\left(\tau,R+M\right)\nonumber\\
& {} &  -\tilde{\psi}\left(\tau,R+M\right)H-H^{T}\tilde{\psi}\left(\tau,R+M\right)\Bigg)x\Bigg]\Bigg) e^{-Tr\left(\tilde{\psi}\left(\tau,R+M\right)x\right)}
\end{eqnarray}
Using equations \eqref{51.24}-\eqref{51.26} in equation \eqref{51.20a} and canceling $f\left(\tau, x\right)$ throughout, we get
\begin{eqnarray}\label{51.27}
-\frac{\partial\tilde{\phi}}{\partial \tau}-Tr\left[\frac{\partial\tilde{\psi}}{\partial \tau}x\right] & = & -Tr\left[\beta Q^{T}Q\tilde{\psi}\left(\tau,R+M\right)\right]-\left(\bar{r}+Tr\left[\left(R+M\right)x\right]\right)+Tr\Bigg[\Bigg(2\tilde{\psi}\left(\tau,R+M\right)\nonumber\\
& {} &  \times Q^{T}Q\tilde{\psi}\left(\tau,R+M\right)-\tilde{\psi}\left(\tau,R+M\right)H-H^{T}\tilde{\psi}\left(\tau,R+M\right)\Bigg)x\Bigg]\Bigg)
\end{eqnarray}
Comparing the terms independent of $x$ and the coefficients of $x$ on both sides of equation \eqref{51.27}, we get the required system of ODEs given in equation \eqref{51.20}. This completes the proof.
\end{proof}

The methodology of solving the system of Riccati equations given in \eqref{51.20} appears in \cite{Fonseca} where the authors propose that matrix Riccati equations can be linearized by doubling the dimension of the problem, Interested readers can also refer to \cite{Grasselli} and \cite{Deelestra}. We state without proof the solution in the following proposition.

\begin{prop}\label{411.9}
The functions $\tilde{\phi}$ and $\tilde{\psi}$ in Theorem \ref{411.8} are given by
\begin{equation}\label{51.200}
\begin{cases}
\tilde{\psi}\left(\tau,R+M\right) = A_{22}^{-1}\left(\tau\right)A_{21}\left(\tau\right) ,\\
\tilde{\phi}\left(\tau, R+M\right) = \frac{\beta}{2}\left(\log\left(\det\left(A_{22}\left(\tau\right)\right)\right)+\tau Tr\left[H^{T}\right]\right).\end{cases}
\end{equation}
where
\begin{equation}\label{51.201}
\begin{pmatrix}
    A_{11}\left(\tau\right) & A_{12}\left(\tau\right) \\
    A_{21}\left(\tau\right) & A_{22}\left(\tau\right)
  \end{pmatrix}
= \exp\left(\tau
\begin{pmatrix}
    H & 2Q^{T}Q \\
    R+M & -H^{T}
  \end{pmatrix}\right)
\end{equation}
\end{prop}
\bigskip

Alternative approaches for the pricing of zero coupon bond under the Wishart short rate model can be found in \cite{Grasselli} and \cite{Gnoatto2}.

\subsubsection{Price of the GAO}
We use Theorem \ref{411.8} and equation \eqref{51.124} to obtain the price of the GAO under the transformed measure $\tilde{Q}$ as
\begin{equation}\label{51.202}
C(0, x, T ) = g \tilde{P}\left(0,T\right){\displaystyle\tilde{E}}\left[\left(\sum_{i=1}^{n-1}e^{-\left(\bar{r}+\bar{\mu}\right)i}e^{-\tilde{\phi}\left(i,R+M\right)-Tr\left[\tilde{\psi}\left(i,R+M\right)X_{T}\right]}-\left(K-1\right)\right)^{+}\right],
\end{equation}
where $\tilde{P}\left(0,T\right)$ is given by equation \eqref{51.20} with $\tau=T$ while $\tilde{\psi}\left(i,R+M\right)$ and $\tilde{\phi}\left(i,R+M\right)$ for $i=1,2,...,n-1$ are given by the system of equations \eqref{51.200} with $\tau=i$.

\subsubsection{Distribution of $X_{T}$}
In order to obtain explicit bounds for the GAO in the Wishart case, we need to obtain the distribution of $X_{T}$ under the transformed measure $\tilde{Q}$. We state this in the following proposition (c.f. \cite{Deelestra} and \cite{Kang} for details).

\begin{prop}\label{411.10}
The dynamics of the Wishart process $X$ defined in equation \eqref{51.17} under the transformed measure $\tilde{Q}$ are given by
\begin{equation}\label{51.203}
dX_{t}=\left(\beta Q^{T}Q+H\left(t\right)X_{t}+X_{t}H\left(t\right)^{T}\right)dt+\sqrt{X_{t}}dW_{t}Q+Q^{T}dW^{T}_{t}\sqrt{X_{t}}, \;t\geq 0,
\end{equation}
where
\begin{equation}\label{51.204}
H\left(t\right)=H-Q^{T}Q\tilde{\psi}\left(\tau,R+M\right),
\end{equation}
$X_{0}=x \in S_{d}^{+}$, $\beta \geq d-1$, $H \in M_{d}$ and $Q \in GL_{d}$. Then
\begin{equation}\label{51.205}
X_{T} \sim \mathscr{W}_{d}\left(\beta,V\left(0\right),V\left(0\right)^{-1}\psi\left(0\right)^{T}x\psi\left(0\right)\right),
\end{equation}
where $\mathscr{W}_{d}$ stands for non-central Wishart Distribution with parameters $d$, $\beta$, $V\left(0\right)$ and $\psi\left(0\right)^{T}x\psi\left(0\right)$ with the last parameter known as non-centrality parameter and is denoted by $\Theta$. Moreover $V\left(t\right)$ and $\psi\left(t\right)$ solve the following system of ODEs
\begin{equation}\label{51.206}
\begin{cases}
\frac{d}{dt}\psi\left(t\right) = -H\left(t\right)^{T}\psi\left(t\right) ,\\
\frac{d}{dt}V\left(t\right) = -\psi\left(t\right)^{T} Q^{T}Q \psi\left(t\right), \end{cases}
\end{equation}
with terminal conditions $\psi\left(T\right)=I_{d}$ and $V\left(T\right)=0$.
\end{prop}
We now state two propositions in context of non-central Wishart Distribution which are very important for the derivation of bounds for the GAO in the Wishart case (c.f. \cite{Pfaffel} and \cite{Gupta})

\begin{prop}\label{411.11}(Laplace Transform of Non-Central Wishart Distribution)
Let $X_{T} \sim \mathscr{W}_{d}\left(\beta,V\left(0\right), \Theta\right)$ with $\Theta = V\left(0\right)^{-1}\psi\left(0\right)^{T}x\psi\left(0\right)$. Then the Laplace transform of $X_{T}$ is given by
\begin{equation}\label{51.207}
\mathscr{L}\left(U\right)=\tilde{E}\left[e^{Tr\left[-UX_{T} \right]}\right]=\det \left(I_{d}+2V\left(0\right)U\right)^{-\frac{\beta}{2}}e^{Tr\left[ -\Theta\left(I_{d}+2V\left(0\right)U\right)^{-1}V\left(0\right)U \right]}
\end{equation}
where $U \in S_{d}^{+}$.
\end{prop}

\begin{prop}\label{411.12}(Characteristic Function of Non-Central Wishart Distribution)
Consider $X_{T} \sim \mathscr{W}_{d}\left(\beta,V\left(0\right), \Theta\right)$ with $\Theta = V\left(0\right)^{-1}\psi\left(0\right)^{T}x\psi\left(0\right)$. Then the Characteristic Function of $X_{T}$ is given by
\begin{equation}\label{51.208}
\phi_{X_{T}}\left(\Lambda\right)=\tilde{E}\left[e^{Tr\left[ i\Lambda X_{T} \right]}\right]=\det \left(I_{d}-2iV\left(0\right)\Lambda\right)^{-\frac{\beta}{2}}e^{Tr\left[ i\Theta\left(I_{d}-2V\left(0\right)\Lambda\right)^{-1}V\left(0\right)\Lambda \right]}
\end{equation}
where $\Lambda \in M_{d}$.
\end{prop}

\subsubsection{The Lower Bound $\mbox{GAOLB}^{\left(WIS\right)}$}
Under the Wishart set up, lower bound $\mbox{GAOLB}$ obtained in equation \eqref{6.10} reduces to  a very neat form. Before arriving at the formula, we define the following notations in the spirit of section 4.
\begin{equation}\label{51.209}
S_{T}^{\left(i\right)}=S_{0}^{\left(i\right)}e^{X_{T}^{\left(i\right)}};\;i=1,2,...,n-1,
\end{equation}
where
\begin{equation}\label{51.210}
S_{0}^{\left(i\right)}=e^{-\left(\left(\bar{r}+\bar{\mu}\right)i+\tilde{\phi}\left(i,R+M\right)\right)}
\end{equation}
and
\begin{equation}\label{51.211}
X_{T}^{\left(i\right)}=-Tr\left[\tilde{\psi}\left(i,R+M\right)X_{T}\right],
\end{equation}
where $\tilde{\psi}\left(i,R+M\right)$ and $\tilde{\phi}\left(i,R+M\right)$ for $i=1,2,...,n-1$ are given by the system of equations \eqref{51.200} with $\tau=i$.  Further, $X_{T}$ has a non-central Wishart distribution with Laplace transform given in equation \eqref{51.207}. This result along with the formula \eqref{6.10a}, the lower bound for the Wishart case manifests itself into the following form.
\begin{eqnarray}\label{51.212}
\mbox{ GAOLB}^{\left(WIS\right)} & = & g \tilde{P}\left(0,T\right)\Bigg({\displaystyle \sum_{i=1}^{n-1}}\Bigg(e^{-\left(\left(\bar{r}+\bar{\mu}\right)i+\tilde{\phi}\left(i,R+M\right)\right)}\det \left(I_{d}+2V\left(0\right)\tilde{\psi}\left(i,R+M\right)\right)^{-\frac{\beta}{2}}\nonumber\\
& {} & {}\times e^{Tr\left[ -\Theta\left(I_{d}+2V\left(0\right)\tilde{\psi}\left(i,R+M\right)\right)^{-1}V\left(0\right)\tilde{\psi}\left(i,R+M\right) \right]}\Bigg)-\left(K-1\right)\Bigg)^{+}
\end{eqnarray}

\subsubsection{The Upper Bound $\mbox{GAOUB}^{\left(WIS\right)}$}
Under the formulation of the assets in the basket under the Wishart case (\eqref{51.209}-\eqref{51.211}), we have
\begin{equation}\label{51.212a}
Y_{0}^{\left(n-1\right)}=-\frac{\left(\bar{r}+\bar{\mu}\right)n}{2}-\frac{1}{n-1} {\displaystyle \sum_{k=1}^{n-1}} \tilde{\phi}\left(k,R+M\right)
\end{equation}
using the definition of log-geometric average $Y_{T}^{\left(n-1\right)}$ given in equation \eqref{20.20} in Section 6. Further, obtaining the upper bound for the GAO in the Wishart set up is a straightforward exercise as one can exploit the upper bound GAO formula under the affine case given in equation \eqref{20.37} by calculating the Laplace transform given in equation \eqref{51.207} such that for $k=1,2,...,n-1$,
\begin{equation}\label{51.213}
\mathscr{L}\left(\tilde{\psi}\left(k,R+M\right)\right)=\det \left(I_{d}+2V\left(0\right)\tilde{\psi}\left(k,R+M\right)\right)^{-\frac{\beta}{2}}e^{Tr\left[ -\Theta\left(I_{d}+2V\left(0\right)\tilde{\psi}\left(k,R+M\right)\right)^{-1}V\left(0\right)\tilde{\psi}\left(k,R+M\right) \right]}
\end{equation}
and calculating $\phi_{X_{T}}\left(-\frac{\left(\eta-i\left(\delta+1\right)\right)}{n-1}\sum_{k=1}^{n-1}\tilde{\psi}\left(k,R+M\right)\right)$ and $\phi_{X_{T}}\left(\frac{i}{n-1}\sum_{k=1}^{n-1}\tilde{\psi}\left(k,R+M\right)\right)$ from the formula \eqref{51.208} by replacing $\Lambda$ by $-\frac{\left(\eta-i\left(\delta+1\right)\right)}{n-1}\sum_{k=1}^{n-1}\tilde{\psi}\left(k,R+M\right)$ and $\frac{i}{n-1}\sum_{k=1}^{n-1}\tilde{\psi}\left(k,R+M\right)$ respectively.

\section{Numerical Results}
Now we investigate the applications of the theory derived in the previous sections. We have successfully obtained a number of lower bounds and an upper bound for Guaranteed Annuity Options in sections 5 and 6. We now test these vis-a-vis the well-known Monte Carlo estimate for the GAO. We carry out this working for a couple of more general affine models. The nomenclature for the bounds has already been specified in sections 5, 6 and 7. In all the examples, we have the following `Contract Specification':
\[
g=11.1\%,\; T=15,\;  n=35;
\]

\subsection{Multi CIR Model}

First we consider a 3-dimensional CIR process $X:=\left(X_{t}\right)_{t\geq 0}$ having independent components $\left(X_{it}\right)_{t\geq 0}$, $i=1,2,3$ (c.f. \cite{Deelestra} for details). We assume the following dynamics for the interest rate process and the mortality process.
\begin{equation}\label{10.1}
r_{t}=\bar{r}+X_{1t}+X_{2t}
\end{equation}
and
\begin{equation}\label{10.2}
\mu_{t}=\bar{\mu}+m_{2}X_{2t}+m_{3}X_{3t},
\end{equation}
where $\bar{r}$, $\bar{\mu}$, $m_{2}$ and $m_{3}$ are constants. We use model specifications similar to \cite{Deelestra} and make a minute alteration in the parameter set. We fix the value of $m_{2}$ and obtain the value of $m_{3}$ such that the expectation of the mortality is fixed to a specified level denoted by $C_{x}\left(T\right)$ which is predicted by e.g. a Gompertz-Makeham model (c.f. \cite{Dickson}) at age $x+T$ for an individual aged x at time 0, i.e.,
\begin{equation}\label{10.3}
\mathbb{E}\left[\mu_{t}\right]=C_{x}\left(T\right),
\end{equation}
Applying expectation on both sides of equation \eqref{10.2} and substituting in \eqref{10.3} we get
\begin{equation}\label{10.4}
\bar{\mu}+m_{2}\mathbb{E}\left[X_{2t}\right]+m_{3}\mathbb{E}\left[X_{3t}\right]=C_{x}\left(T\right),
\end{equation}
where as $X_{it};\;i=1,2,3$ is obtained using the Stochastic Differential Equation (SDE) given by \eqref{50.1}, we have
\begin{equation}\label{10.5}
\mathbb{E}\left[X_{it}\right]=X_{i,0}e^{-k_{i}T}+\theta_{i}\left(1-e^{-k_{i}T}\right).
\end{equation}
Using our contract specifications outlined in the beginning of this section we fix the expected value in \eqref{10.3} to the level $C_{50}\left(15\right)=0.0125$. A very good discussion in regards to the validity of the model to be used for mortality appears in \cite{Deelestra}. In fact this model was completely calibrated in \cite{Chiarella}.

Using the set up defined by equations \eqref{10.1}-\eqref{10.2}, the linear pairwise correlation between $\left(r_{t}\right)_{t\geq 0}$ and $\left(\mu_{t}\right)_{t\geq 0}$, denoted by $\rho_{t}$ forms a stochastic process given by
\begin{equation}\label{10.6}
\rho_{t}=\frac{m_{2}\sigma_{2}^{2}X_{2t}}{\sqrt{\sigma_{1}^{2}X_{1t}+\sigma_{2}^{2}X_{2t}}\sqrt{m_{2}^2\sigma_{2}^{2}X_{2t}+m_{2}^2\sigma_{2}^{2}X_{2t}}}.
\end{equation}
We vary the value of $m_{2}$ and therefore obtain the value of $m_{3}$ using equation \eqref{10.3} and this finally yields the value of $\rho$. Further in line with \cite{Deelestra}, we make the following parameter specifications
\[
\bar{r}=-0.12332,\;\;\bar{\mu}=0
\]

\begin{table}[ht]
\centering      
\begin{tabular}{c c c c c c c c}  
\hline\hline\\[0.01ex]                     
CIR process &   &  Parameters &   &  & \\ [1ex] 
\hline\hline\\[0.5ex]                  
$X_{1}$ & $k_{1}=0.3731$ & $\theta_{1}=0.074484$ & $\sigma_{1}=0.0452$ & $X_{1,0}=0.0510234$ \\[1ex]
$X_{2}$ & $k_{2}=0.011\;$ & $\theta_{1}=0.245455$ & $\sigma_{2}=0.0368$ & $X_{2,0}=0.0890707$ \\[1ex]
$X_{3}$ & $k_{3}=0.01\;\;$ & $\theta_{1}=0.0013\;\;$ & $\sigma_{3}=0.0015$ & $X_{3,0}=0.0004\;\;\;\;$ \\[1ex]
\hline\hline     
\end{tabular}
\label{table1}  
\caption{Parameter Values for the 3-dimensional CIR process} 
\end{table}

Table 2 depicts the lower bound, the upper bound and the Monte Carlo estimate of the GAO price for different values of $m_{2}$ and therefore for different values of the initial pairwise linear correlation coefficient $\rho_{0}$. We find that an increase in the value of $\rho_{0}$ enhances the value of the GAO. The lower bound is extremely sharp. On the other hand, upper bound is slightly wider. The results of Table 2 are portrayed in Figures 1-3.

Figure 2 reflects that the relative difference ($=\frac{|bound-MC|}{MC}$) between the upper bound and the benchmark Monte Carlo estimate decreases with an increase in the correlation between mortality and interest rate while the relative difference for the lower bound almost remains constant with varying $\rho_{0}$. On the other hand, figure 3 depicts the absolute difference between the Monte Carlo estimate of the GAO price and the derived bounds which remain more or less constant. The lower bound fares much better than $GAOUB$. Finally figure 4 shows the price bounds and in fact the lower bound stick completely camouflages with that of the MC estimate which is a testimony to the tightness of the lower bound.

\begin{table}[ht]
\centering      
\begin{tabular}{|c||c||c|c|c|c|c|c|}  
\hline
$\;\;m_{2}\;\;$ & $\;\rho\;$ & $\;\;\;\;\;\mbox{GAOLB}^{\left(MCIR\right)}\;\;\;\;$  & $\;\;MC\;\;\;\;\;\;\;\;\;\;\;$ & $\;\;\;\;\;\mbox{GAOUB}^{\left(MCIR\right)}\;\;$ \\ \hline\hline
-0.300 & -0.570960646515027 & 0.153351236437789 & 0.153431631010533 & 0.216286630652776 \\
-0.100 & -0.460513730466363 & 0.181641413947461 & 0.181871723226662 & 0.243710313225013 \\
-0.070 & -0.403426257094426 & 0.187186872445969	& 0.187285214852833 & 0.249173899703122 \\
-0.060 & -0.376271648827787	& 0.189122188373390	& 0.189373949402726	& 0.251083730739797 \\
-0.050 & -0.343007585286942	& 0.191102351580502	& 0.191474297920361	& 0.253039217047040 \\
-0.040 & -0.301756813619030	& 0.193128263182051	& 0.195421232722993	& 0.255041205164633 \\
-0.030 & -0.250041147986350	& 0.195200853300304	& 0.195132243321684	& 0.257090572307459 \\
-0.020 & -0.184739400604580	& 0.197321081986930	& 0.197531098187496	& 0.259188227324353 \\
-0.010 & -0.102346730178820	& 0.199489940182500	& 0.199619257104038	& 0.261335111674878 \\
-0.001 & -0.011167160239806	& 0.201484335480591	& 0.201710195921424	& 0.263310203859562 \\
0.000 & 0.000000000000000 & 0.201708450715130 &	0.201879045816498 &	0.263532200435103 \\
0.001 & 0.011370596893292 & 0.201933073002533 &	0.202090425152612 &	0.263754709122691 \\
0.010 & 0.122142590872118 & 0.203977669339908 &	0.204292134604299 &	0.265780503352825 \\
0.020 & 0.257493768936871 & 0.206298685820891 &	0.206369996912367 &	0.268081065972310 \\
0.030 &	0.391761086281179 & 0.208672625057373 &	0.208709896009824 &	0.270434970824946 \\
0.040 &	0.508145173072700 & 0.211100648256358 &	0.211180180724315 &	0.272843338639536 \\
0.050 &	0.596334605305204 & 0.213583954153270 &	0.213584231985838 &	0.275307329501738 \\
0.060 &	0.656025897318996 & 0.216123780282872 &	0.216228415988778 &	0.277828143936307 \\
0.070 &	0.693071640464574 & 0.218721404302618 &	0.218840241843838 &	0.280407024005732 \\
0.100 &	0.730953349866014 & 0.226874471461256 &	0.226934772658478 &	0.288505131181583 \\
\hline
\end{tabular}
\label{table30}  
\caption{Lower Bounds and Upper Bound $\mbox{GAOUB}$ for Guaranteed Annuity Option under the MCIR Model with partial parameter choice in accordance with \cite{Deelestra}. MC Simulations: 50000}
\end{table}

\subsection{Wishart Model}

As a final step we test our bounds in the backdrop of the celebrated Wishart model for mortality and interest rate. The functional form of the model for the two aforesaid risks has been detailed in equations \eqref{51.18a}-\eqref{51.18b}. To present the application of our methodology we stick to a 2-dimensional Wishart process, i.e., $d=2$ due to the fact that higher dimensional Wishart processes are difficult to implement. The law for the underlying process $X$ governing the mortality and interest rate processes has been outlined in equation \eqref{51.17}. We consider partial choice of the parameter set in accordance with \cite{Deelestra}. For all examples considered below, let
\[
\beta=3,\;\;\bar{r}=0.04,\;\;\bar{\mu}=0,
\]
\begin{equation}\label{10.7}
H = \begin{pmatrix}
    -0.5 & 0.4 \\
    0.007 & -0.008
  \end{pmatrix},\;\;
M = \begin{pmatrix}
    0 & 0 \\
    0 & 1
  \end{pmatrix},\;\;
R = \begin{pmatrix}
    1 & 0 \\
    0 & 0
  \end{pmatrix}
\end{equation}
in equations, \eqref{51.17} and \eqref{51.18a}-\eqref{51.18b}. In light of this data, the stochastic correlation between $\left(r_{t}\right)_{t\geq 0}$ and $\left(\mu_{t}\right)_{t\geq 0}$, denoted by $\rho_{t}$ forms a stochastic process given by
\begin{equation}\label{10.9}
\rho_{t}=\frac{\left(Q_{11}Q_{12}+Q_{22}Q_{21}\right)X_{t}^{12}}{\sqrt{\left(Q_{11}^2+Q_{21}^{2}\right)X_{t}^{11}\left(Q_{22}^2+Q_{12}^{2}\right)X_{t}^{22}}}.
\end{equation}
As is evident from \eqref{10.9}, using a Wishart formulation for underlying process $X$ produces a more richer dependence structure for the underlying risks than was available under the multidimensional CIR case. This calls for carrying out a more sophisticated sensitivity analysis in regards to the involved parameters. In the same spirit as \cite{Deelestra}, we carry out a two-fold testing
\begin{itemize}
\item the first one by varying the off-diagonal elements of the initial Wishart process $X_{0}$ and investigating the impact on the prices of the GAO,
\item the second one by experimenting with the off-diagonal elements of the matrix $Q$.
\end{itemize}
In each case, we compute the bounds and compare them with the benchmark Monte Carlo value which is computed using 20000 simulations. For stability checks in relation to the expected values of the interest rate and mortality intensity w.r.t. varying correlation, interested readers can refer to \cite{Deelestra}.

\subsubsection{Effect of a Change in Initial Value $X_{0}$}
In order to see the behaviour of the price bounds for the GAO price vis-a-vis change in the initial value of the Wishart process, we experiment with two cases:
\begin{itemize}
\item Negative off-diagonal elements in the volatility matrix Q
\item Positive off-diagonal elements in the volatility matrix Q.
\end{itemize}
\bigskip

\textbf{Example 1}. In this case we consider the following Wishart process:
\begin{equation}\label{10.10}
Q = \begin{pmatrix}
    0.06 & -0.0006 \\
    -0.06 & 0.006
  \end{pmatrix},\;\;
X_{0} = \begin{pmatrix}
    0.01 & X_{0}^{12} \\
    X_{0}^{12} & 0.001
  \end{pmatrix}.\;\;
\end{equation}
Table 3 portrays the lower bound, the upper bound and the Monte Carlo estimate of the GAO price for different values of $X_{0}^{12}$ and therefore for different values of the initial pairwise linear correlation coefficient $\rho_{0}$. We find that an increase in the value of $\rho_{0}$ enhances the value of the GAO in a fashion similar to the one shown for the Multi-CIR set up in Table 3. In this case both the lower and the upper bounds show close proximity to the GAO value.
\bigskip

\begin{table}[ht]
\centering      
\begin{tabular}{|c||c||c|c|c|c|c|c|}  
\hline
$\;\;X_{0}^{12}\;\;$ & $\;\rho\;$ & $\;\;\;\;\;\mbox{GAOLB}^{\left(WIS\right)}\;\;\;\;$  & $\;\;MC\;\;\;\;\;\;\;\;\;\;\;$ & $\;\;\;\;\;\mbox{GAOUB}^{\left(WIS\right)}\;\;$ \\ \hline\hline
-0.003	& 0.734240363158475	& 0.241898614923743	& 0.241247798732840 & 0.241898616247735 \\
-0.002	& 0.489493575438983	& 0.241133565561902	& 0.240529742517039	& 0.241133567256078 \\
-0.0015	& 0.367120181579237	& 0.240751892681841	& 0.239890712272120	& 0.240751894570155 \\
-0.0005	& 0.122373393859746	& 0.239990246464251	& 0.239141473598451	& 0.239990248759807 \\
0	& 0.000000000000000	& 0.239610271506445	& 0.238621509824004	& 0.239610274015476 \\
0.0005	& -0.122373393859746	& 0.239230860904335	& 0.238198077279364 & 0.239230863633664 \\
0.0015	& -0.367120181579237	& 0.238473729539084	& 0.237679950746197	& 0.238473732730283 \\
0.002	& -0.489493575438983	& 0.238096007164703	& 0.237331879397777	& 0.238096010597887 \\
0.003	& -0.734240363158475	& 0.237342245012764	& 0.236699850447918	& 0.237342248953105 \\

\hline
\end{tabular}
\label{table3}  
\caption{Lower Bounds and Upper Bound $\mbox{GAOUB}$ for Guaranteed Annuity Option under the Wishart Model Example 1 with parameter choice in accordance with \cite{Deelestra}. MC Simulations: 20000}
\end{table}

\textbf{Example 2}. In the second investigation, we consider the following Wishart process:
\begin{equation}\label{10.11}
Q = \begin{pmatrix}
    0.06 & 0.00001 \\
    0.0002 & 0.006
  \end{pmatrix},\;\;
X_{0} = \begin{pmatrix}
    0.01 & X_{0}^{12} \\
    X_{0}^{12} & 0.001
  \end{pmatrix}.\;\;
\end{equation}
As can be seen, in this example, we consider positive off-diagonal elements for the matrix $Q$. Table 4 portrays the lower bound, the upper bound and the Monte Carlo estimate of the GAO price for different values of $X_{0}^{12}$ and therefore for different values of the initial pairwise linear correlation coefficient $\rho_{0}$. The results obtained present a sharp contrast to those obtained in Table 3 and the value of the GAO price and the corresponding bounds begin to drop as the value of $\rho_{0}$ is increases. Both the bounds continue to perform well even on this occasion.

A good justification of the behaviour of the GAO price in the first two examples (also see  \cite{Deelestra}) vis-a-vis the values of $X_{0}^{12}$ can be provided by noting that under dynamics of the Wishart process (\eqref{51.17}) the positive factors swell on an average when the initial value $X_{0}^{12}$ increases. Moreover, for the aforementioned parameter choice, the models for mortality and interest rate for $t \geq 0$ are given as
\begin{equation}\label{10.12}
r_{t}=0.04+X_{t}^{11}
\end{equation}
and
\begin{equation}\label{10.13}
\mu_{t}=X_{t}^{11}.
\end{equation}
Now, it is clear from the formula for GAO price given in equation \eqref{51.202}, that the exponential term containing $r_{t}$ and $\mu_{t}$ decays when $X_{0}^{12}$ increases and this causes the GAO price and corresponding bounds to diminish when $X_{0}^{12}$ soars.

\begin{table}[ht]
\centering      
\begin{tabular}{|c||c||c|c|c|c|c|c|}  
\hline
$\;\;X_{0}^{12}\;\;$ & $\;\rho\;$ & $\;\;\;\;\;\mbox{GAOLB}^{\left(WIS\right)}\;\;\;\;$  & $\;\;MC\;\;\;\;\;\;\;\;\;\;\;$ & $\;\;\;\;\;\mbox{GAOUB}^{\left(WIS\right)}\;\;$ \\ \hline\hline
-0.003 &	-0.004743383550130 &	0.332948404889575 &	0.341196353690094 &	0.332948737923275 \\
-0.002	& -0.003162255700087 &	0.331667762094902 &	0.340868095614857 &	0.331668129236460 \\
-0.0015 &	-0.002371691775065 &	0.331029148831226 &	0.339651861654315 &	0.331029534152954 \\
-0.0005 &	-0.000790563925022 &	0.329755328754714 &	0.339133498851769 &	0.329755752815905 \\
0 &	0.000000000000000 &	0.329120118025352 &	0.338246665341653 &	0.329120562698337 \\
0.0005	& 0.000790563925022 &	0.328486037563152 &	0.337667153845205 &	0.328486503712013 \\
0.0015	& 0.002371691775065 &	0.327221259643971 &	0.336913730200477 &	0.327221771448801 \\
0.002	& 0.003162255700087 &	0.326590558297188 &	0.336554330647556 &	0.326591094339721 \\
0.003	& 0.004743383550130 &	0.325332521129667 &	0.335045167150194 &	0.325333108616048 \\

\hline
\end{tabular}
\label{table4}  
\caption{Lower Bounds and Upper Bound $\mbox{GAOUB}$ for Guaranteed Annuity Option under the Wishart Model Example 2. MC Simulations: 20000}
\end{table}

\subsubsection{Effect of a Change in Volatility Matrix $Q$}
We now carry out an experiment to vary the off-diagonal elements of the volatility matrix $Q$ which we assume to be symmetric while specifying the initial value $X_{0}$ of the Wishart process.
\bigskip

\textbf{Example 3}. Here the Wishart process is as follows:
\begin{equation}\label{10.14}
Q = \begin{pmatrix}
    0.06 & Q_{12} \\
    Q_{12} & 0.006
  \end{pmatrix},\;\;
X_{0} = \begin{pmatrix}
    0.01 & 0.001 \\
    0.001 & 0.001
  \end{pmatrix}.\;\;
\end{equation}

Table 5 depicts the lower bound, the upper bound and the Monte Carlo estimate of the GAO price for different values of $Q_{12}$ and therefore for different values of the initial pairwise linear correlation coefficient $\rho_{0}$. The results obtained show that the value of the GAO price and the corresponding bounds do not show a monotone behaviour in respect of the linear correlation between mortality and interest rate risks. The tightness of the bounds around the Monte Carlo estimate still remains intact. These observations are echoed in Figure 7. In addition Figure 5 reflects that the relative difference $\left(=\frac{|bound-MC|}{MC}\right)$ between the lower bound and the benchmark Monte Carlo estimate increases with an increase in the correlation $\rho_{0}$ between mortality and interest rate. For example looking at table 5, we see that the relative difference for $\mbox{GAOLB}$ increases from a meagre $0.2\%$ for $\rho_{0}=-0.3$ to about $7.7\%$ for a $\rho_{0}=0.3$. However, under the same set, the relative difference between the estimated GAO price and the upper bound increases and then there is a switch at $\rho_{0}=0.3$ and this gap begins to diminish. The last observation is also seen in Figure 6 for the absolute difference between the bounds and the MC estimate of GAO price.

The reason for this behaviour of the GAO price lies in the structure of the matrix $Q^{T}Q$ (also see \cite{Deelestra}). It is clear that the diagonal elements of $Q^{T}Q$ increase with a rise in the absolute value of $Q_{12}$. A glance at the law of the Wishart process given in equation \eqref{51.17} and equations \eqref{10.12}-\eqref{10.13} brings out the fact that the drift and in particular the long term value of the positive factors of the Wishart process and in turn the drift of mortality and interest rate process is an increasing function of the absolute value of $Q_{12}$. Thus an upward rise in the value of $Q_{12}$ will enhance the positive factors. As a result, it is evident from equation \eqref{51.202} describing the GAO price in the Wishart case, that the exponential term containing $r_{t}$ and $\mu_{t}$ decreases when $Q_{12}$ moves away from zero and this causes the GAO price and corresponding bounds to diminish.

Overall our numerical experiments provide strong evidence in support of the extremely adequate performance of our proposed bounds.

\begin{table}[ht]
\centering      
\begin{tabular}{|c||c||c|c|c|c|c|c|}  
\hline
$\;\;Q_{12}\;\;$ & $\;\rho\;$ & $\;\;\;\;\;\mbox{GAOLB}^{\left(WIS\right)}\;\;\;\;$  & $\;\;MC\;\;\;\;\;\;\;\;\;\;\;$ & $\;\;\;\;\;\mbox{GAOUB}^{\left(WIS\right)}\;\;$ \\ \hline\hline
-0.01 &	-0.294220967543866 & 0.290016256883993 & 0.290601398401997 & 0.290593286187411 \\
-0.006 & -0.244746787719492 & 0.331837945818948 & 0.332421093218907 & 0.331843140669134 \\
-0.002 & -0.109938939767707 & 0.339526376457815 & 0.344143066326585 & 0.339526466816062 \\
0.002 &	0.109938939767707 &	0.308919593324378 &	0.322579113504993 &	0.308928343737340 \\
0.006 &	0.244746787719492 &	0.257040019380241 &	0.274376988651895 &	0.257891897298705 \\
0.01 & 0.294220967543866 & 0.196440417823759 &	0.212744888444368 &	0.204994244625801 \\
\hline
\end{tabular}
\label{table5}  
\caption{Lower Bounds and Upper Bound $\mbox{GAOUB}$ for Guaranteed Annuity Option under the Wishart Model Example 3 with parameter choice in accordance with \cite{Deelestra}. MC Simulations: 20000}
\end{table}

\begin{figure}[h]
\centering
\includegraphics[scale=0.8]{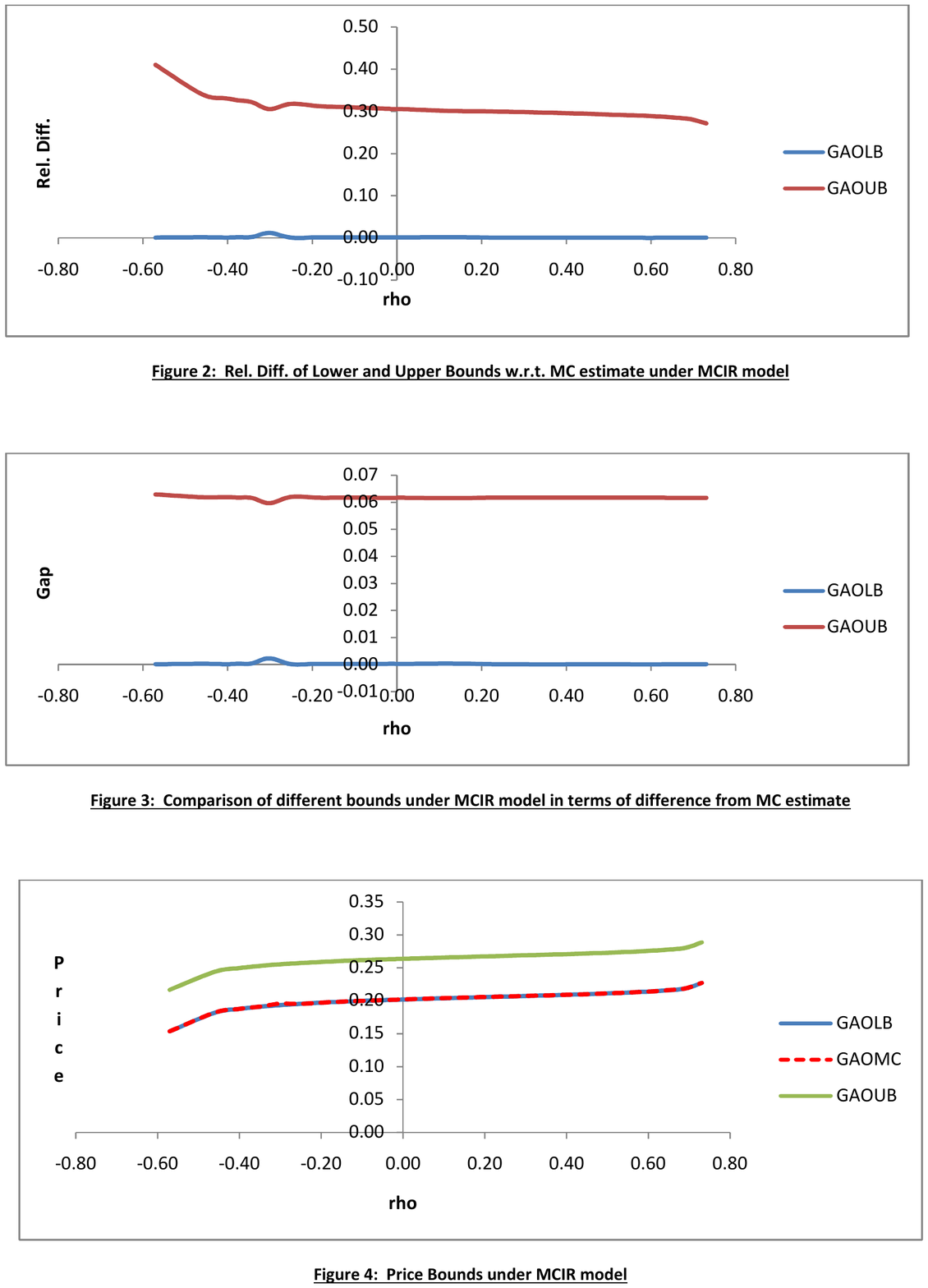}
\end{figure}

\begin{figure}[h]
\centering
\includegraphics[scale=0.8]{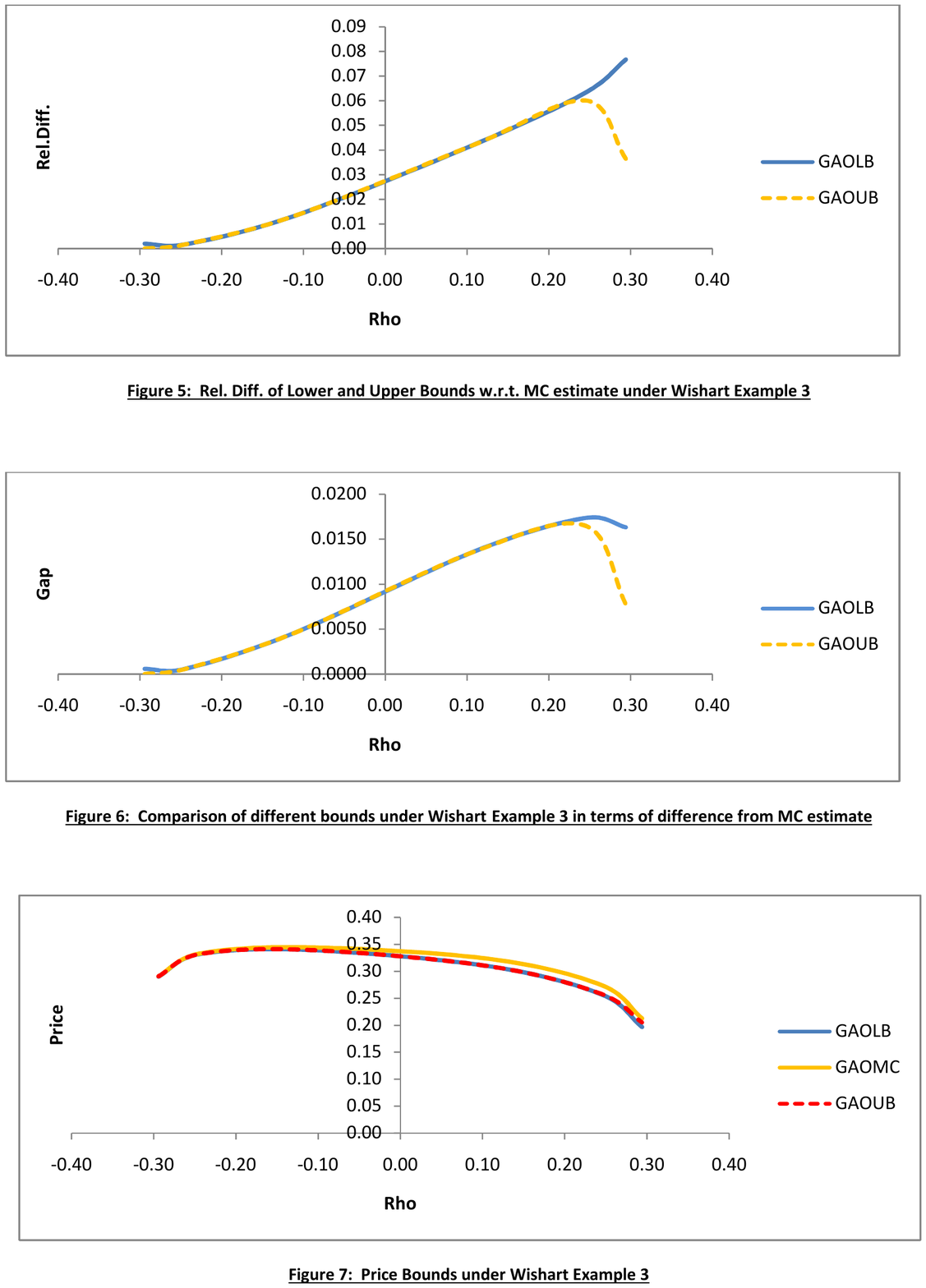}
\end{figure}

\subsection{Computational Speed of the Bounds}
We summarize the time consumed in computation of the bounds and the Monte Carlo estimate in the following table. Further, these observations are portrayed in the figure that follows.

\begin{table}[ht]

\centering      
\begin{tabular}{c c c c c c c}  
\hline\hline\\                
        &       &       & Number of & Simulations & for & Monte Carlo \\
Example & GAOLB & GAOUB & 1000 & 10000 & 20000 & 50000 \\
\hline\hline
MCIR & 0 & 1 & 44 & 352 & 696 & 1800 \\
Wishart 1 & 0 & 1 & 47 & 369 & 749 & 2100 \\       
Wishart 2 & 0 & 1 & 49 & 379 & 757 & 2200 \\
Wishart 3 & 0 & 1 & 43 & 359 & 724 & 2000\\
\hline\hline     
\end{tabular}
\label{table100}  
\caption{Time taken in seconds for Bounds and Simulations} 
\end{table}

\begin{figure}[H]
    \centering
      \includegraphics[clip, trim=0.7cm 15.2cm 0.9cm 4cm, width=1.00\textwidth]{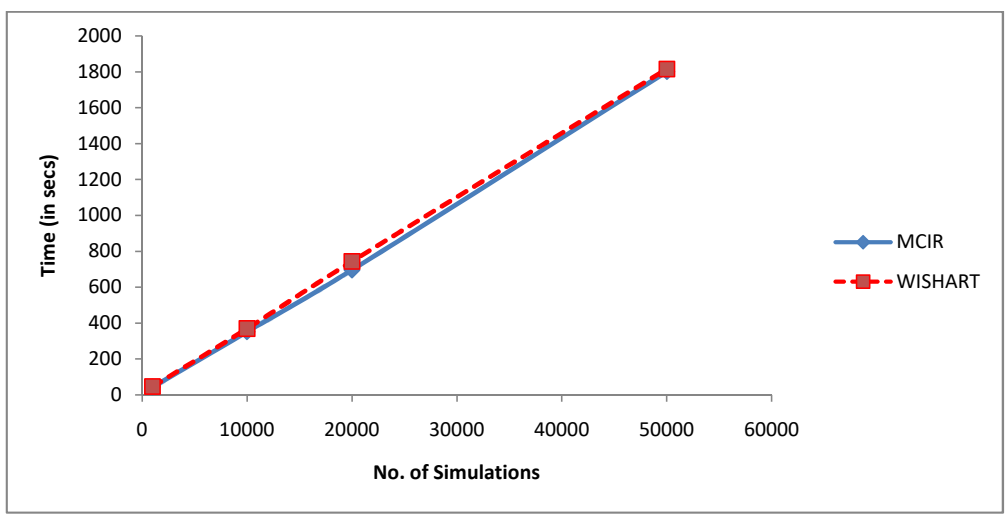}
    \caption{The CPU time (seconds) for MCIR and Wishart (average for 3 cases)}
    \label{fig2}
\end{figure}
All computations in Section 8 are carried out on a personal laptop with Intel(R) Core(TM) i5 CPU-M450 at 2.40 GHz and a RAM of 4.00 GB.

\section{Conclusions}
We have derived some very general bounds for the valuation of GAO's under the assumption of a prevailing correlation between mortality and interest rate risk. These bounds serve as a useful tool for financial institutions which are striving hard to find methodologies that offer efficient pricing of longevity linked securities. The techniques used in this paper are successful in circumventing the issue of dealing with sums of a large number of correlated variables. Moreover they are extremely useful in reducing the burden of dealing with cumbersome stochastic processes.



The most successful finding of this research is that in the affine case, both the lower and the upper bound depend on the properties of the distribution of the random variables connected to the transformed stochastic processes underlying mortality and interest rate. Moreover the lower bound manifests itself in form of Laplace transform of the underlying random variable while the upper bound reveals itself in the form of the associated characteristic function. Both of these tools are the most conveniently obtained vital statistics for any distribution. The most satisfying aspect is that we need to work in one dimension, in contrast to what would have been atleast a 34-dimensional set up, assuming that a person lives atleast 100 years making $n=35$.


Another feather in the cap of the bounds is their computational speed. As indicated in the previous section, the Monte Carlo method is extremely slow for large number of simulations in case of sophisticated models. As a result given the same time budget, Monte Carlo estimates are deemed to be extremely inaccurate. Moreover for highly sophisticated multivariate distributions like non-central Wishart, generating random samples generally involves complex algorithms, which are not inbuilt in libraries of packages such as MATLAB. It is indeed very satisfactory that our lower bound takes just 0.133 or 0.192 seconds on an average to execute in the MCIR and Wishart case while an average of about 0.286 or 0.659 seconds are required by the upper bound respectively in the two cases. In other words pricing of complex GAOs can be done in no time.

Last but not the least the fact that the upper bound performs much better in the case of Wishart is a noteworthy observation since the Wishart model is much more intricate in terms of gauging the mesh of correlation between mortality and interest rate.

The sensitivity analysis done in this article reiterates the fact that it is not possible to explain the value of a GAO completely in terms on the initial pairwise linear correlation between mortality and interest rate risks as highlighted by the Wishart model (c.f. \cite{Deelestra} for earlier work in this direction). This finding sends alarm signals for the risk management in the presence of an unknown dependence as various scenarios are possible.

If the prices of a GAO increase with the (initial) linear correlation coefficient as in the multi-CIR model or the Wishart specifications in Example 1, then the most risk-averse methodology when pricing a GAO would be to take the linear correlation coefficient equal to unity. This will protect the seller from an awkward scenario of underestimation of the GAO price in the event of a high correlation. However, Example 2 of the Wishart case, presents an opposite scenario where prices decrease with increasing initial linear correlation and therefore, risk-adverse seller would adopt the opposite rule in that situation. Example 3 in the Wishart case portrays prices which are not monotone
with respect to the correlation, but which seem to lead to the highest prices for zero correlation. Therefore, in this situation, choosing zero correlation might be the appropriate risk-averse choice. In fact the Wishart model comes across the most versatile model presenting all possible dependence scenarios.

The methodology proposed in this paper is extremely flexible and can be easily extended to value other insurance products such as indexed annuities or to instruments with option embedded features such as equity-linked annuities, equity-indexed annuities and variable annuities.



\appendix
\section{Appendix A}
\begin{defn} \label{411.2} \textbf{\emph{Affine Process}} A time-homogeneous Markov process $X$ relative to some filtration $\left(\mathscr{F}_{s}\right)$ and with state space $\left(D,\mathscr{D}\right)$ (augmented by $\Delta$) is called affine if\\

(i) it is stochastically continuous, that is, $\lim_{s \rightarrow t}p_{s}\left(x, \cdot\right)=p_{t}\left(x, \cdot\right)$ for all $t\geq 0$ and $x \in D$, and
\\[12pt]
(ii) its Fourier-Laplace transform has exponential affine dependence on the initial state. This means that there exist functions $\phi:\mathbb{R}_{+}\times S_{d}^{+}\rightarrow \mathbb{R}_{+}$ and $\psi:\mathbb{R}_{+}\times S_{d}^{+}\rightarrow S_{d}^{+}$ such that
\begin{equation}\label{51.2}
\mathbb{E}_{x}\left[e^{\langle u,X_{t}\rangle}\right]=P_{t}e^{\langle u,x\rangle}=\int_{D} e^{\langle u,\xi\rangle}p_{t}\left(x, d\xi \right)=e^{-\phi\left(t,u\right)-\langle \psi\left(t,u\right),x\rangle},
\end{equation}
for all $x\in D$ and $\left(t,u\right)\in \mathbb{R}_{+}\times \mathbb{R}_{d}$
\end{defn}

\begin{defn} \label{411.5} \textbf{\emph{Truncation Function}} Let $\chi:S_{d}\rightarrow S_{d}$ be some bounded continuous truncation function with $\chi\left(\xi\right)=\xi$ in the neighborhood of $0$. An admissible parameter set given by $\left(\alpha, b, \beta^{ij}, c, \gamma, m, \mu\right)$ associated with $\chi$ consists of:
\begin{itemize}
\item a linear diffusion coefficient
\begin{equation}\label{51.3a}
\alpha \in S_{d}^{+},
\end{equation}
\item a constant drift term
\begin{equation}\label{51.3b}
b \succeq \left(d-1\right)\alpha,
\end{equation}
\item a constant killing rate term
\begin{equation}\label{51.3c}
c \in \mathbb{R}^{+},
\end{equation}
\item a linear killing rate coefficient
\begin{equation}\label{51.3d}
\gamma \in S_{d}^{+},
\end{equation}
\item a constant jump term: a Borel measure $m$ on $S_{d}^{+}\setminus \{0\}$ satisfying
\begin{equation}\label{51.3e}
\int_{S_{d}^{+}\setminus \{0\}}\left(\parallel \xi \parallel \wedge 1\right)m\left(d\xi\right) < \infty,
\end{equation}
\item a linear jump coefficient: a $d \times d$ matrix $\mu=\left(\mu_{ij}\right)$ of finite signed measures on $S_{d}^{+}\setminus \{0\}$ such that $\mu\left(E\right)\in S_{d}^{+}$ for all $E\in \mathscr{B}\left(S_{d}^{+}\setminus \{0\}\right)$ and the kernel
\begin{equation}\label{51.3f}
M\left(x,d\xi\right):=\frac{\langle x,\mu\left(d\xi\right)\rangle}{\parallel \xi \parallel^{2} \wedge 1}
\end{equation}
satisfies
\begin{equation}\label{51.3g}
\int_{S_{d}^{+}\setminus \{0\}}\langle \chi\left(\xi\right),u\rangle M\left(x,d\xi\right) < \infty \;\;\mbox{for all } x,u \in S_{d}^{+} \mbox{with } \langle x,u\rangle=0,
\end{equation}
\item a linear drift coefficient: a family $\beta^{ij}=\beta^{ji}\in S_{d}$ such that the linear map $B:S_{d}\rightarrow S_{d}$ of the form
\begin{equation}\label{51.3h}
B\left(x\right)=\sum_{i,j}\beta^{ij}x_{ij}
\end{equation}
satisfies
\begin{equation}\label{51.3i}
\langle B\left(x\right),u\rangle-\int_{S_{d}^{+}\setminus \{0\}}\langle \chi\left(\xi\right),u\rangle M\left(x,d\xi\right) \geq 0 \;\;\mbox{for all } x,u \in S_{d}^{+} \mbox{with } \langle x,u\rangle=0.
\end{equation}
\end{itemize}
\end{defn}

\begin{defn} \label{411.4} \textbf{\emph{Generator}} For an affine process $X$ taking values in $S_{d}^{+} \subset S_{d}$ the infinitesimal generator is defined as
\begin{equation}\label{51.4}
\mathscr{A}f\left(x\right)=\lim_{t\rightarrow 0^{+}}\frac{\mathbb{E}\left[f\left(X_{t}^{x}\right)\right]-f\left(x\right)}{t}\;\mbox{for }x\in S_{d}^{+},\;f\in\mathscr{C}^{2}\left(S_{d}, \mathbb{R}_{d}\right) \mbox{with bounded derivatives}.
\end{equation}
\end{defn}

\end{document}